\documentclass[11pt]{article}
\usepackage{rheaj}
\usepackage{setspace}
\usepackage{xparse}

\newcommand{\mypara}[1]{\medskip \noindent {\bf #1}}

\newcommand{\racke}{R\"{a}cke\xspace}

\newcommand{\flow}{\textnormal{flow}}
\renewcommand{\eps}{\varepsilon}
\newcommand{\droute}{\calD^{(1)}}
\newcommand{\bablp}{\textbf{LP-BB}}
\newcommand{\hclp}{\textbf{LP-HC}}
\newcommand{\junctionlp}{\textbf{LP-Junction}}
\newcommand{\hop}{\textnormal{hop}}
\newcommand{\LG}{\textnormal{LARGE}}
\newcommand{\SM}{\textnormal{SMALL}}
\newcommand{\diam}{\textnormal{diam}}
\newcommand{\density}{\textnormal{density}}
\newcommand{\treeround}{\textnormal{TreeRounding}}

\begin{document}
\title{Approximation Algorithms for Hop Constrained and Buy-at-Bulk Network Design
  via Hop Constrained Oblivious Routing}
\author{Chandra Chekuri \thanks{Dept. of Computer Science, Univ. of
  Illinois, Urbana-Champaign, Urbana, IL 61801. {\tt
  chekuri@illinois.edu}. Supported in part by NSF grants CCF-1910149 and CCF-1907937.}
\and
Rhea Jain\thanks{Dept. of Computer Science, Univ. of Illinois,
  Urbana-Champaign, Urbana, IL 61801. {\tt rheaj3@illinois.edu}. Supported in part by NSF grant CCF-1907937.}
}
\date{\today}
\maketitle

\begin{abstract}
  We consider two-cost network design models in which edges of the
  input graph have an associated \emph{cost} and \emph{length}.
  We build upon recent advances in hop-constrained oblivious routing
  to obtain two sets of results.

  We address multicommodity buy-at-bulk network design in the
  nonuniform setting. Existing
  poly-logarithmic approximations are based on the junction tree
  approach \cite{chks09,kortsarz_nutov11}.  We obtain a new
  polylogarithmic approximation via a natural LP relaxation. This
  establishes an upper bound on its integrality gap and affirmatively
  answers an open question raised in \cite{chks09}.  The rounding is
  based on recent results in hop-constrained oblivious routing
  \cite{hop_congestion21}, and this technique yields a polylogarithmic
  approximation in more general settings such as set connectivity.
  Our algorithm for buy-at-bulk network design is based on an LP-based
  reduction to $h$-hop constrained network design for which we obtain
  LP-based bicriteria approximation algorithms. 
  
  We also consider a fault-tolerant version of $h$-hop constrained
  network design where one wants to design a low-cost network to
  guarantee short paths between a given set of source-sink pairs even
  when $k-1$ edges can fail. This model has
  been considered in network design \cite{GouveiaL17,GouveiaML18,ArslanJL20} 
  but no approximation algorithms were known.  We obtain polylogarithmic
  bicriteria approximation algorithms for the single-source setting
  for any fixed $k$. We build upon the
  single-source algorithm and the junction-tree approach to obtain an
  approximation algorithm for the multicommodity setting when at most
  one edge can fail.
\end{abstract}

\section{Introduction}
\label{sec:intro}
Network design is a fundamental area of research in algorithms
touching upon several related fields, including combinatorial
optimization, graph theory, and operations research.
A  canonical problem is Steiner Forest: given a graph
$G = (V, E)$ with non-negative edge costs $c: E \to \R_+$ and terminal
pairs $s_i, t_i \in V$ for $i \in [r]$, the goal is to find $F \subseteq E$ that connects
all $s_i$-$t_i$ pairs while minimizing the total cost of $F$. Steiner
Tree is the special case when $s_i = s$ for
all $i$.  
Steiner Tree is referred to as \emph{single-source} problem,
while Steiner forest is a \emph{multicommodity} problem.
Both are NP-Hard and APX-hard to approximate, and extensively studied.
Steiner Forest has a $2$-approximation \cite{Jain01} and Steiner Tree has a $\ln 4 +
\eps$-approximation \cite{ByrkaGRS13}.
In this paper we address a
class of \emph{two-cost} network design problems. The input to these
problems is a graph $G=(V,E)$ where each edge has a
non-negative cost $c(e)$, and a non-negative length $\ell(e)$;  $c(e)$
represents a fixed cost and $\ell(e)$ represents
a hop constraint or routing cost. The goal is to choose a low-cost subgraph of $G$ to satisfy connectivity
and/or routing requirements for some given set of source-sink pairs.
The two-cost model is important due to its ability to model a number 
of fundamental problems. Despite many advances over the years, there
are still several problems that are important for both theory and
applications, but are not well-understood. 
We consider two
sets of problems, formally described below.
We obtain several new approximation algorithms and resolve
an open problem from \cite{chks09}.

\mypara{Buy-at-Bulk:} In buy-at-bulk network design, the goal is to
design a low-cost network to support routing demands between
given source-sink pairs. 
The cost of buying capacity on an edge to support the flow routed on
it is typically a subadditive function; this arises naturally in
telecommunication networks and other settings in which costs exhibit
economies of scale.  Buy-at-bulk is an important problem in practice
and has been influential in the approximation algorithms literature
since its formal introduction in \cite{awerbuch97}; see Section 
\ref{sec:related_work} for details.  We study the multicommodity
version, denoted MC-BaB. The input consists of a graph $G = (V, E)$, a
set of $r$ demand pairs $\{s_i, t_i\}_{i \in [r]}$ with demand
$\delta(i)$ each, and a monotone sub-additive cost function
$f_e: \R^+ \to \R^+$ for each edge $e \in E$.  We need to route
$\delta(i)$ units of flow from $s_i$ to $t_i$ for each $i \in
[r]$. Given a routing, its cost is $\sum_{e \in E} f_e(x_e)$ where
$x_e$ is the total amount of flow sent on edge $x_e$. The goal is to
find a routing of minimum cost. We focus on the \emph{nonuniform}
setting, where each edge has its own cost function $f_e$, and refer to
the case where $f_e$ is the same for all edges ($f_e = c_e \cdot f$
for some function $f$) as the \emph{uniform} problem. We call an
instance \emph{single-source} if all pairs have the same source $s$,
i.e. $\exists s \in V$ such that $s_i = s$ for all $i \in [r]$.  One
can simplify the problem, by losing an approximation factor of $2$
\cite{az98}, wherein we can assume that each $f_e$ has a simple
piecewise linear form: $f_e(x) = c(e) + \ell(e)\cdot x$; we call
$c(e)$ the \emph{cost} and $\ell(e)$ the \emph{length}.
Even though buy-at-bulk is naturally defined via routing flows,
the two-cost model allows one to recast it in a different light by
considering fixed costs and hop lengths in the aggregate. In
particular, the problem now is to choose a set of edges $F \subseteq
E$ where the objective function value of $F$ is defined as
$\sum_{e \in F} c(e) + \sum_{i \in [r]} \delta(i) \ell_F(s_i,t_i)$,
where $\ell_F(s_i,t_i)$ is the shortest path length between
$s_i$ and $t_i$ in the graph induced by $F$. Notice that
the choice of routing has been made implicit in the objective.
The uniform versions of the problem can be handled via
metric embedding techniques \cite{awerbuch97}, however,
the nonuniform problem has been particularly challenging.
The best known approximation ratios for MC-BaB are
$O(\log^4 r)$ in the nonuniform setting \cite{chks09} (one can obtain
a log factor improvement when demands are poly-bounded
\cite{kortsarz_nutov11}) and $O(\log r)$ in the uniform setting
\cite{awerbuch97,frt_tree,bartal_tree,gmm01}.  In single source
buy-at-bulk, there exists an $O(\log r)$-approximation in
the nonuniform setting \cite{mmp08,ckn01} and $O(1)$ for uniform
\cite{gmm01,talwar02,gmr03,grandoni_italiano06}.  
The integrality gap of a natural LP relaxation for MC-BaB has been
unresolved for over fifteen years \cite{chks09}. We also note that
MC-BaB is hard to within $\Omega(\log^c n)$-factor even in uniform
settings (see Section \ref{sec:related_work}).

\mypara{Hop-constrained network design:} The goal is to design
low-cost networks in which source-sink pairs are connected by paths
with few edges. We say that a path has hop-length $h$ if it has $h$
edges -- this corresponds to path length with $\ell(e) = 1$ for all
$e$. Such a hop-length constraint is natural in many telecommunication
networks and is extensively studied in theory and practice, see Section
\ref{sec:related_work} for details. Here, we study $h$-Hop Constrained Steiner Forest
($h$-HCSF).  The input is a graph $G = (V, E)$ with non-negative edge
costs $c:E \to \R_+$ and $r$ terminal pairs $s_i, t_i \in V$.  The
goal is to find $F \subseteq E$ minimizing
$c(F) := \sum_{e \in F} c(e)$ such that for all $i \in [r]$, there
exists an $s_i$-$t_i$ path in $F$ of hop-length at most $h$.
Although hop-constrained problems had admitted bicriteria approximations in
single-source and spanning settings, multicommodity versions such as $h$-HCSF
were harder to tackle, and until recently, no non-trivial
approximation algorithms (even bicriteria) were known. 
In the past few years, this
barrier has been overcome through the use of probabilistic tree
embeddings with hop constraints
\cite{hop_distance21,hop_congestion21,filtser22}. These results allow
us to project the graph onto a tree, solve the problem on the
tree, and project back to the original graph with distances preserved
up to small factors.
This has led to the development of several
\emph{bicriteria} approximation algorithms, in which the returned
subgraph $F$ connects each terminal pair with a path of length at most
$\polylog(n) \cdot h$ and the total cost of $F$ is at most
$\polylog(n)$ times the optimal \cite{hop_distance21,filtser22}.
We note that hop-constrained probabilistic tree embeddings pose several new 
challenges beyond those of traditional metric 
embeddings \cite{bartal_tree,frt_tree} since the trees are partial (do not
contain all vertices). 
Most
recently, Ghaffari, Haeupler and Zuzic gave a congestion-based
hop-constrained tree embedding \cite{hop_congestion21}, providing a
hop-constrained analog to \racke's seminal tree-based oblivious
routing result \cite{Racke08}.

MC-BaB and $h$-HCSF are related two-cost network design
problems. While $h$-HCSF imposes a strict hop-constraint on the paths,
MC-BaB penalizes the hop-length for the pairs in the objective
function. The relation allows some results for one to be translated to the
other with care. 

\mypara{Fault-tolerance:} Fault-tolerant network design has been
studied in a wide variety of settings and has numerous practical
applications. Via Menger's theorem, two nodes are $k$-edge-connected if
there exist $k$ edge-disjoint paths between them, or, equivalently, if
they remain connected despite the failure of any $k-1$ edges.
Survivable Network Design Problem (SNDP) is a central problem in his
context. The input is similar to Steiner Forest; in addition each $s_it_i$ pair
specifies an integer requirement $k_i$ and the goal is to find a
min-cost subgraph of the given graph $G$ in which each $s_it_i$ is
$k_i$-edge-connected.  Jain's seminal work \cite{Jain01} obtained a
$2$-approximation for this problem. In the hop-constrained setting,
the situation is more complicated. The corresponding variant of
Menger's theorem does not hold (see \cite{GouveiaL17} and Figure
\ref{fig:hc_fault_example}).  Thus, there are two natural
higher-connectivity generalizations of hop-constrained network design.

The first version is the focus of this paper and is the fault-based
generalization of hop-constrained network design, introduced in
\cite{GouveiaL17}. We say $u, v \in V$ are
\emph{$(h,k)$-hop-connected} in $G$ if there exists a path using at
most $h$ edges in $E \setminus Q$ for all $Q \subseteq E$, $|Q| < k$.
For instance, $s$ and $t$ are $(h, 2)$-hop-connected in Figure
\ref{fig:hc_fault_example}.  We define the
\emph{$(h,k)$-Fault-Tolerant Hop Constrained Network Design} problem,
denoted $(h,k)$-Fault-HCND: given $r$ terminal pairs $s_i, t_i \in V$
with associated connectivity requirements $k_i$, the goal is to find
$F \subseteq E$ minimizing $c(F)$ such that for all $i \in [r]$,
$s_i, t_i$ are $(h, k_i)$-hop-connected. We let
$k = \max_{i \in [r]} k_i$.

\begin{figure}[bht]
  \begin{center}
    \includegraphics[width=\linewidth]{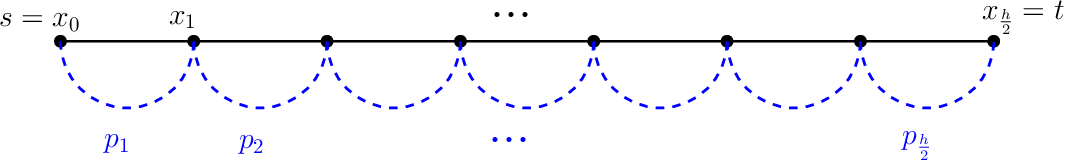}
  \end{center}
  \caption{Each $p_i$ has $\hop(p_i) = h/2$. $s$ and $t$ 
  are connected by a path of hop-length $h$ despite the failure of 
  any one edge. However, one needs hop-length $\Omega(h^2)$ to 
  obtain two edge disjoint $s$-$t$ paths.}
  \label{fig:hc_fault_example}
\end{figure}

The second version is the Hop-Constrained Survivable Network Design
problem (HC-SNDP) problem: given a hop constraint $h$ and terminal pairs $s_i$, $t_i$ with
connectivity requirements $k_i$, a feasible solution contains $k_i$
edge-disjoint $s_i$-$t_i$ paths, \emph{each} of hop length at most $h$. The
goal is to find a min-cost feasible subgraph. In the example in Figure
\ref{fig:hc_fault_example}, $s$ and $t$ are $(h,2)$-hop-connected
but any two $s$-$t$ disjoint paths must have at least one path of length $\Omega(h^2)$.

As noted in \cite{GouveiaL17,GouveiaML18}, the two versions are both
interesting and relevant from a telecommunication network design point
of view.  While there exist several algorithms via IP-solvers for both
versions of hop-constrained problem (see Section \ref{sec:related_work}), 
there are no known efficient approximation algorithms. 

Now we discuss buy-at-bulk in the fault-tolerant setting, an
important practical problem first studied in approximation by
Antonakapolous et al. \cite{acsz11}. The goal is to find a
min-cost subgraph that allows the demands to be routed even under
failures which requires routing each demand along disjoint paths.  It
has proved to be a challenging problem. For instance, even in the
single-source nonuniform setting, poly-logarithmic approximation
algorithms are not known for protecting against a single edge failure!
Progress has been made for some special cases (see Section \ref{sec:related_work}).  One
can show that bicriteria algorithms for HC-SNDP would imply algorithms
for fault-tolerant nonuniform Buy-at-Bulk and vice-versa.  Our current
techniques do not seem adequate to address these problems in their
full generality even though there are natural LP relaxations. In this
version we focus on the first version of fault-tolerant
hop-constrained network design which is also challenging.

\mypara{Set Connectivity:}
There are several
problems where we seek connectivity between pairs of \emph{sets}.
In such problems, the input consists of set pairs $S_i, T_i \subseteq V$ 
instead of terminal pairs $s_i, t_i \in V$. The goal is to find paths 
connecting $S_i$ to $T_i$ for each $i \in [r]$. This problem was 
first introduced in the single source setting, known as Group Steiner Tree, 
and has been influential in network design (see Section \ref{sec:related_work}).
Recent progress has resulted in algorithms for the hop-constrained version 
\cite{hop_distance21,filtser22}
and also for the fault-tolerant version \cite{cllz22}.
One can naturally extend the Buy-at-Bulk 
problem to this more general setting, which we consider in this work.

\subsection{Overview of Results}
\label{sec:results}
We consider polynomial-time approximation algorithms. 
For Buy-at-Bulk problems we consider the standard
approximation ratio. For $h$-Hop-Constrained Network
Design problems we consider bi-criteria $(\alpha, \beta)$-approximations
where $\alpha$ is the approximation for the cost of the output solution and
$\beta$ is the violation in the hop constraint.
In discussing results we let $\opt_{I}$ denote the value of an optimal
(integral) solution, and $\opt_{LP}$ to denote the value of an optimum
fractional solution to an underlying LP relaxation (see Section \ref{sec:prelim}). 

Our first set of results address the multicommodity settings of
two-cost network design problems. Theorem \ref{thm:bab_main} resolves
an open question posed by \cite{chks09} by proving a polylogarithmic
approximation for multicommodity Buy-at-Bulk with respect to the LP
relaxation. The resulting algorithm provides a new approach to the problem.

\begin{theorem}
\label{thm:bab_main}
  There is a randomized $O(\log D\log^3 n \log r)$-approximation for
  multicommodity Buy-at-Bulk with respect to $\opt_{LP}$, where 
  $D = \max_{i \in [r]} \delta(i)$ is the maximum demand. This 
  extends to Buy-at-Bulk Set Connectivity and the approximation ratio 
  is $O(\log D\log^7(nr))$. 
\end{theorem}

There is an $O(\log^4 r)$-approximation for MC-BaB with respect to the
optimal integral solution \cite{chks09}; our result matches this up to
an $O(\log D)$ factor. Removing this factor is a direction for future
research, see Remark \ref{rem:bab_logD}. No previous approximation
algorithm was known for Buy-at-Bulk Set Connectivity.

Theorem \ref{thm:hc_main} below describes a polylogarithmic bicriteria 
approximation for $h$-Hop Constrained Steiner Forest with respect to the 
optimal fractional solution, and also for Set Connectivity version
(referred to as $h$-HCSC).

\begin{theorem}
\label{thm:hc_main}
  There exist polylogarithmic bicriteria approximation algorithms for 
  hop-constrained network design problems with respect to $\opt_{LP}$. 
  In $h$-Hop Constrained Steiner Forest, we obtain cost factor 
  $\alpha = O(\log^2n \log r)$ and hop factor $\beta = O(\log^3 n)$. 
  In $h$-Hop Constrained Set Connectivity, we obtain cost factor 
  $O(\log^5 n \log r)$ and hop factor $\beta = O(\log^3 n)$. 
\end{theorem}

There exist bicriteria approximations for both $h$-HCSC and
$h$-HCSF with respect to optimal integral solutions via partial tree
embeddings \cite{hop_distance21}, which were later improved through
the use of clan embeddings \cite{filtser22}. Our results match or
improve on those in \cite{hop_distance21}; see Section 
\ref{sec:related_work} for details. 
In addition to inherent interest
and other applications, an important reason for proving
Theorem \ref{thm:hc_main} is to derive Theorem \ref{thm:bab_main}
via a reduction (see Section~\ref{sec:bab}). 

Our next set of results are on hop-constrained network design in the
\emph{fault-tolerant} setting.  There were previously no known
approximation algorithms for this problem. 
Theorems \ref{thm:hcfault_main} obtains bicriteria approximation for
the single-source setting for any fixed number of failures.

\begin{theorem}
\label{thm:hcfault_main}
  There is a randomized 
  $(O(k^3\log^6n), O(k^k \log^3n))$-approximation for \emph{single-source}
  $(h,k)$-Fault-HCND with 
  respect to $\opt_{LP}$ and runs in time $n^{O(k)}$.
\end{theorem}

We then consider the multicommodity setting and prove Theorem
\ref{thm:junction_main} for $k=2$ via a reduction to single-source. 
We note that the result is with respect to the optimal integral solution.

\begin{theorem}
\label{thm:junction_main}
  There is a randomized $(O(\log^6n \log^2 r), O(\log^3 n \log r))$-approximation 
  for multicommodity $(h,2)$-Fault-HCND with respect to $\opt_{I}$.
\end{theorem}

There are several well-known advantages for considering LP based
algorithms even when there exist combinatorial algorithms. 
First, working with fractional solutions allows us to
extend our results to otherwise complex settings. For example, without
much additional overhead, we are able to extend our results to set
connectivity and \emph{prize collecting} versions of the problems. 
In prize collecting, each demand pair is
associated with penalty $\pi_i \ge 0$. The goal is to find a
subset $S \subseteq [k]$ and minimize the cost of connecting the
demand pairs in $S$ plus the penalty for the pairs not in $S$. It
is well-known that LP-based approximation algorithms for a
network design problem can be used to obtain a corresponding
approximation for its prize-collecting variant with only an
additional constant factor loss (\cite{bgsw93} and follow-up
work). Thus our results imply the same asymptotic approximation
ratios for the prize collecting variants. It is sometimes possible
to improve LP-based algorithms for prize-collecting versions via
more involved methods as was done recently for Steiner Forest 
\cite{AhmadiGHJM-prize}, 
however, it requires substantial problem-specific machinery.

Another important motivation is to address higher connectivity
variants.  Though the extension is nontrivial, we are able to use
ideas from the LP-based multicommodity algorithms to obtain algorithms
for the fault-tolerant variants of hop-constrained network design. We
believe that this work provides a stepping stone to handle the
challenges that remain in addressing 
fault-tolerant buy-at-bulk
as well as extending Theorem~\ref{thm:junction_main} for $k > 2$.
As we remarked earlier, all the fault-tolerant versions have natural
LP relaxations, however, bounding their integrality gap is
challenging. 

\subsection{Overview of Techniques}
\label{sec:techniques}

The proofs of Theorems \ref{thm:bab_main} and \ref{thm:hc_main} are
inherently connected, as multicommodity buy-at-bulk and
hop-constrained network design essentially reduce to each other if one
is willing to lose polylogarithmic factors in the approximation ratios
(see Remark \ref{rem:bab_reduction_direction}). 
Such a connection between the
problems has been pointed out in the past in the single-source setting
with respect to integral solutions \cite{hks09}; we extend this to an
LP based relation with a bit of technical work. 
Our main contribution is to prove Theorem \ref{thm:hc_main}.
We note that existing approximation algorithms for $h$-HCSF 
use distance-based embeddings \cite{hop_distance21}. Unlike single-cost 
network design problems, these hop-constrained tree embeddings 
cannot be directly used to argue about fractional solutions.
This is because the embeddings are \emph{partial}
trees that do not include all vertices, so there is no natural
projection of a fractional solution on $G$ to a solution on $T$.
Instead, we leverage 
a connection to the recent
\emph{congestion-based} hop-constrained partial tree embeddings
\cite{hop_congestion21}. 
We use this tool by solving the LP
relaxation for $h$-HCSF and using the fractional $x_e$ values to
define a capacitated graph $G'$ from which we sample trees.
This allows us to solve the relevant problems
on the tree and project the solution back to the input graph. 

Our second set of results is for the fault-tolerant setting of
hop-constrained problems. We approach the problems here
via \emph{augmentation}, where we start with an initial partial
solution and repeatedly augment to increase the connectivity of
terminal pairs whose requirement is not yet satisfied.
It is not clear how congestion-based tree embeddings are useful
for higher connectivity problems since polylogarithmically many paths
in the input graph can map to the same path in the tree embedding.
To overcome this,  we rely on a recent powerful approach of Chen et al. for Survivable
Set Connectivity \cite{cllz22} which has also found some extensions in
\cite{cj23}. The framework provides a clever way to
assign capacities to edges, avoiding the aforementioned difficulty
while ensuring that $O(\polylog n)$ rounds of an oblivious dependent
tree-rounding process are adequate to obtain a feasible
solution. Unlike
\cite{cllz22,cj23} who applied this framework with \racke's congestion
based tree embeddings, we need to use hop-constrained version
\cite{hop_congestion21}. Unfortunately, the framework does not
generalize in a nice way due to the additional complexity of hop-constraints.
We were able to overcome this difficulty in the single-source setting by analyzing the
diameter of components as the augmentation algorithm proceeds. The
multicommodity setting poses non-trivial challenges; thus we restrict
our attention to $k=2$ which is often relevant in practice. We use the
\emph{junction-tree} technique developed in \cite{chks09} for
buy-at-bulk network design and later extended by \cite{acsz11} for
fault-tolerant buy-at-bulk when $k=2$. This approach reduces the
multicommodity problem to the single-source setting. However, the
argument relies on the integral optimum solution and hence
Theorem~\ref{thm:junction_main} is only with respect to $\opt_{I}$.
Our key contribution for the multicommodity setting with $k=2$ is
showing the existence of a good junction structure in the fault-tolerant 
hop-constrained setting.

\subsection{Related Work}
\label{sec:related_work}

\subsubsection{Group Steiner Tree and Set Connectivity}

Group Steiner tree (GST), the single source variant of Set Connectivity, was first
introduced by Reich and Widmayer \cite{ReichW89}. Garg, Konjevod and Ravi studied
approximation algorithms for the problem and provided a randomized rounding
scheme, giving an $O(\log n \log r)$-approximation when the input graph $G$ is
a tree \cite{GargKR98}. This was extended to an $O(\log^2n \log r)$-approximation
for general graphs via tree embeddings \cite{bartal_tree,frt_tree}. These
results are essentially tight; there is no efficient $\log^{2-\eps}(r)$
approximation for GST \cite{HalperinK03}, and there is an
$\tilde \Omega(\log n \log r)$ lower bound on the integrality gap of GST
when the input graph is a tree \cite{HalperinKKSW07}.
Approximation algorithms when allowing for
slightly less efficient algorithms (running in quasi polynomial time)
have been studied, see \cite{ChekuriP05,GrandoniLL19,GhugeN22}.
Multicommodity Set Connectivity was first introduced by Alon et al.
\cite{aaabn06} in the online
setting. The first polylogarithmic approximation ratio for offline
Set Connectivity was given
by Chekuri, Even, Gupta, and Segev \cite{cegs11}.

In the higher connectivity setting, we are given an additional parameter $k$
and the goal is to connect each $S_i$-$T_i$ pair with $k$ edge-disjoint paths.
Gupta et al. showed that
Survivable Group Steiner Tree has a polylogarithmic approximation when
$k = 2$ \cite{gkr10}. Chalermsook, Grandoni,
and Laehkanukit \cite{CGL15} studied the higher connectivity version
of Set Connectivity. They gave a bicriteria approximation algorithm, where
the solution only contains $k / \beta$ edge-disjoint paths between each
pair and costs at most $\alpha$ times the optimal, for 
$\alpha, \beta \in \poly \log (n)$.
This approximation is via \racke congestion-based tree embeddings and
Group Steiner Tree rounding ideas from \cite{GargKR98}. In recent work by
Chen et al. \cite{cllz22}, similar ideas are used to give the first
true polylogarithmic approximation for Survivable Set Connectivity.

\subsubsection{Buy-at-Bulk Network Design}

Buy-at-bulk network design was introduced to the algorithms literature
by Salman et al. \cite{batb97} in 1997 in the single-source setting;
the multicommodity setting was first considered by Awerbuch and Azar 
\cite{awerbuch97}. The model was first introduced to capture 
multiple cable types with different capacities, in which the cost per
unit of bandwidth decreases as the capacity increases, giving us the 
general subadditive cost function we consider in this work.
This problem had been previously considered in practical 
settings and operations research; see 
\cite{GoldsteinR71,BienstockG96,BienstockG96} for a few such examples. 
Since then, it has been
extensively studied in a variety of different settings (uniform vs
non-uniform, node costs, special cases such as rent-or-buy, in the
online setting, and so on).  In this section, we limit our discussion
to relevant results in undirected graphs for the offline edge-cost
Buy-at-Bulk problem.

We begin with a discussion of the uniform setting: when each edge has
the same cost function $f_e$. In the multicommodity setting, Awerbuch
and Azar gave an $O(\log n)$ via probabilistic tree embeddings
\cite{awerbuch97,bartal_tree,frt_tree}; this bound improves to
$O(\log r)$ via a refinement
of the distortion in tree embeddings \cite{gnr10}. In the
single-source setting, the first constant factor approximation
algorithm was given by Guha, Meyerson, and Munagala \cite{gmm01},
later improved by Talwar~\cite{talwar02}, Gupta, Kumar, and
Roughgarden \cite{gmr03}, and finally Grandoni and Italiano
\cite{grandoni_italiano06}.  Andrews showed an
$\Omega(\log^{1/4 - \eps}(n))$ hardness of approximation factor for
multicommodity buy-at-bulk \cite{andrews04}.

In the nonuniform setting, Meyerson, Munagala, and Plotkin gave a
randomized $O(\log r)$-approximation \cite{mmp08} in the single-source
setting.  Chekuri, Khanna, and Naor \cite{ckn01} obtained a
deterministic algorithm by derandomizing the algorithm in \cite{mmp08}
via an LP relaxation which also established that the LP's integrality
gap is $O(\log r)$. The first nontrivial approximation for the general
multicommodity case was a simple randomized greedy algorithm given by
Charikar and Karagiozova \cite{bab_random_greedy} with an
approximation factor of $\exp(O(\sqrt{\ln n \ln \ln n}))$. Chekuri et
al. gave an $O(\log^4 r)$-approximation for nonuniform multicommodity
buy-at-bulk \cite{chks09}.  They also give a simple greedy
combinatorial $O(\log^3 r \log D)$-approximation.  In the case that
$D = \poly(n)$, Kortsarz and Nutov improved the approximation to
$O(\log^3 n)$ \cite{kortsarz_nutov11}.  On hardness of approximation,
multicommodity has an $\Omega(\log^{1/2 - \eps}(n))$ lower bound
\cite{andrews04} and single source has an $\Omega(\log \log n)$ lower
bound \cite{cgns08}.

We briefly discuss Buy-at-Bulk with protection.  Recall that the goal is
to route flow between demand pairs along $k$ edge-disjoint or
node-disjoint routes so that the routing is robust to edge/node
failures. Note that node-connectivity generalizes
edge-connectivity. This protected setting was first introduced by
Antonakapolous et al., who give an $O(1)$-approximation when $k = 2$ in
the single cable, node-connectivity setting \cite{acsz11}.  They also
show when $k = 2$ that an $\alpha$-approximation with respect to the LP 
for single-source implies an $O(\alpha \log^3 n)$-approximation for 
multicommodity via junction arguments, both
in the uniform and nonuniform settings. Chekuri and Korula subsequently
gave two results in the single-source node-connectivity setting:
$(\log n^{O(b)})$-approximation for $k = 2$ with $b$ cables and
$(2^{O(\sqrt n)})$-approximation for any fixed $k$ in the nonuniform
case \cite{chekuri_korula08}.  Finally, in the edge-connectivity
setting, Gupta, Krishnaswamy, and Ravi gave an
$O(\log^2n)$-approximation for the multicommodity uniform setting with
$k = 2$ \cite{gkr10}.  Obtaining a polylogarithmic approximation
algorithm for $k \geq 3$ in the uniform setting and $k \geq 2$ in
general remains an important open problem.

\subsubsection{Hop-Constrained Network Design}
Hop constrained network design was first introduced by Balakrishnan
and Altinkemer \cite{ba92}. It has since been shown that hop
constraints have many applications, including easier traffic
management \cite{acm81,monma_sheng86}, faster communication
\cite{deboeck_fortz18,akgun_tansel11}, improved tolerance to failure
\cite{woolston_albin88,raj12}, and more
\cite{lcm99,gpsv03,dahl98}. Unfortunately, hop constrained network
design problems have been shown to be hard; even MST with hop
constraints has no $o(\log n)$-approximation in polynomial time
\cite{bkp01}.  The majority of the approximation algorithms and IP
formulations for hop constrained network design problems have been
limited to simple connectivity problems, such as MST
\cite{dahl98,akgun_tansel11,pirkul_soni03,althaus05,kls05,ravi94,marathe98},
Steiner Tree \cite{kp97,kp06}, and $k$-Steiner Tree \cite{hks09,ks11}.
One reason for this limitation is the hardness of approximating more
complex problems. For example, hop-constrained Steiner Forest has no
polynomial time $o(2^{\log^{1-\eps}(n)})$-approximation \cite{dkr16}.
Note that the hardness results are when the hop constraints are
strict.  Much less is known on the hardness when hop constraints can
be relaxed by constant or poly-logarithmic factors.
Natural generalizations of Steiner Tree, including
Steiner Forest and Set Connectivity, were not studied in the
hop-constrained setting until recently.

Recent progress on hop-constrained tree embeddings has led to
significant progress in many hop-constrained network design
problems. This work was initiated by Haeupler, Hershkowitz, and Zuzic
\cite{hop_distance21}. They extended the line of work on
distance-based probabilistic tree embeddings (see
\cite{bartal_tree,frt_tree}) to the hop-constrained setting. In
particular, they show that although hop-constrained distances are
inapproximable by metrics, one can approximate hop-constrained
distances with partial trees that contain only a fraction of nodes of
the original graph. This led to the development of bicriteria
approximations for several offline problems ($k$-Steiner Tree, Relaxed
$k$-Steiner Tree, Group Steiner Tree, Group Steiner Forest), as well
as some online and oblivious problems. Note that bicriteria
approximations are necessary due to the hardness result mentioned
earlier. For Hop-Constrained Steiner Forest, they obtain a cost 
factor $\alpha = O(\log^3 n)$ and hop factor $\beta = O(\log^3 n)$ with 
respect to the optimal integral solution. For Hop-Constrained Set 
Connectivity, they obtain a cost factor $\alpha = O(\log ^6 n \log r)$ 
and hop factor $O(\log^3 n)$. 
Many of these results were then improved by Filtser \cite{filtser22}
via \emph{clan}-embeddings, in which embeddings contain multiple
copies of each vertex in the original graph. 
In particular for HCSF they improve the cost factor to $\tilde O(\log^2 n)$,
hop factor to $\tilde O(\log^2 n)$ and 
for HCSC they improve the cost factor to $\alpha = \tilde O(\log^4 n \log k)$, 
hop factor $\beta = \tilde O(\log^2 n)$.
Most recently, Ghaffari,
Haeupler, and Zuzic ~\cite{hop_congestion21} extended their
hop-constrained tree embedding results to give an oblivious routing
scheme with low congestion as well as distance dilation. This is
analogous to \racke's work with congestion-based tree embeddings
\cite{Racke08} without hop constraints.

Hop-constrained network design in the fault-tolerant model has not been 
considered in approximation before this work; however, there is 
a fair amount of literature on algorithms via IP-solvers. These models
have been shown to be of importance in several practical settings, including 
telecommunication networks \cite{Diarrassouba_Mahjoub_Almudahka_2024} and
transportation shipments \cite{BalakrishnanK17,Lhomme15}. 
Recall that 
we consider two generalizations of hop-constraints to higher connectivity. 
The first is Hop-Constrained Survivable Network Design (HC-SNDP), in which 
the goal is to route flow between terminal pairs along $k$ edge-disjoint 
paths, each of bounded hop-length. This problem was first considered in the 
single-pair setting \cite{HuygensMP04}; this is essentially min-cost flow 
if one allows for $O(k)$ hop distortion. There is extensive literature 
on IP formulations in the single-pair setting which we omit here. 
In the past decade, the single source setting \cite{MahjoubSU13} 
and multicommodity setting 
\cite{DiarrassoubaGMGP16,MahjoubPSU19,Diarrassouba_Mahjoub_Almudahka_2024,Botton13} 
have been studied as well. The second generalization is Fault-Tolerant 
Hop Constrained Network Design (Fault-HCND), in which the goal is to 
buy a subgraph that contains a path of bounded hop length between each 
set of terminal pairs despite the failure of any $k-1$ edges. This problem
was recently introduced by \cite{GouveiaL17} under the name 
``Network Design Problem with Vulnerability Constraints'' (NDPVC), 
where they provide several IP formulations for the $k=2$ case (protecting 
against a single edge failure). 
This work was improved in \cite{GouveiaML18} and extended to $k \geq 3$ 
by \cite{ArslanJL20}.

\mypara{Organization:} We describe the LP relaxations and relevant 
background on tools used in the paper in Section \ref{sec:prelim}. 
We prove Theorem \ref{thm:bab_main} by describing a 
reduction from Hop-Constrained Network Design to 
Buy-at-Bulk in the multicommodity setting in Section \ref{sec:bab}.
We then obtain polylogarithmic approximations for 
Hop-Constrained Network Design problems in Section 
\ref{sec:hc}. We discuss the fault-tolerant hop-constrained 
problem (Theorems \ref{thm:hcfault_main} and 
\ref{thm:junction_main}) in Sections \ref{sec:hc_fault}.
\section{Preliminaries}
\label{sec:prelim}

\subsection{Linear Programming Relaxations}

Following \cite{chks09}, we define path-based LP relaxations for both Buy-at-Bulk 
(see \bablp) and hop-constrained network design (see \hclp). 
For $i \in [r]$, let $\calP_i$ denote the set of all simple $s_i$-$t_i$ paths. 
For the hop-constrained setting, let $\hop(p)$ denote the hop-length of 
any path $p$, and let $\calP^h_i$ denote the set of all simple $s_i$-$t_i$ paths 
with hop-length at most $h$. In both relaxations, we define variables 
$x_e \in [0,1]$ for $e \in E$ and $f_p \in [0,1]$ for 
$p \in \bigcup_{i \in [r]} \calP_i$ in \bablp\ and 
$p \in \bigcup_{i \in [r]} \calP_i^h$ in \hclp; these are indicators 
for whether or not an edge $e$ or path $p$ is used. 
\hclp\ can be extended to the fault-tolerant setting, as described in Section 
\ref{sec:hc_fault}. 
The relaxations are described below. 

\vspace{1mm}
\begin{minipage}{0.45\textwidth}
  \begin{align*}
    &(\bablp) \\
    \min &\sum_{e \in E} c(e)x_e + 
    \sum_{i \in [r]} \delta(i) \sum_{p \in \calP_i} \ell(p) f_p \\
    \text{s.t.} &\sum_{p \in \calP_i} f_p = 1 \qquad \qquad \qquad \forall i \in [r] \\
    &\sum_{p \in \calP_i, e \in p} f_p \leq x_e \quad \forall e \in E, \forall i \in [r] \\
    &x_e, f_p \geq 0
  \end{align*}
\end{minipage}%
\vrule{}
\begin{minipage}{0.45\textwidth}
  \begin{align*}
    &(\hclp) \\
    \min &\sum_{e \in E} c(e)x_e \\
    \text{s.t.} &\sum_{p \in \calP^h_i} f_p = 1 &\forall i \in [r] \\
    &\sum_{p \in \calP^h_i, e \in p} f_p \leq x_e &\forall e \in E, \forall i \in [r] \\
    &x_e, f_p \geq 0
  \end{align*}
\end{minipage}

\begin{remark}
\label{rem:lp_separation}
  \hclp~and \bablp~can each be solved in polynomial time via 
  separation oracle on their respective duals.
\end{remark} 

\begin{lemma}
\label{lem:lp_assumption1}
	We assume that given any solution $(x, f)$ to \bablp~or \hclp~that 
	$x_e \geq \frac 1 {m^2}$ by increasing the approximation ratio by 
	at most $1 + 1/(m-1)$. 
\end{lemma}
\begin{proof}
	Let $(x, f)$ be an optimal solution to \bablp. Set $x_e = 0$ for all 
	$x_e < \frac 1 {m^2}$, $f_p = 0$ for all $p$ such that $e \in p$ for 
	some $x_e = 0$, and multiply all remaining $x_e$ and $f_p$ values by 
	$1 + 1/(m-1)$. It is easy to verify that the second constraint remains satisfied, 
	as both $x_e$ and $f_p$ values are scaled by the same amount. For the first, 
	note that after setting $f_p = 0$ if $x_e < \frac 1 {m^2}$, 
	$\sum_{p \in \calP_i} f_p$ decreases by at most 
	$\sum_{e: x_e < m^2} \sum_{p \in \calP_i, e \in p} f_p \leq 
	\sum_{e: x_e < m^2} x_e < \frac 1 m$. Then, we scale all remaining values 
	up by $1 + 1/(m-1)$, so the final sum is at least 
	$\sum_{p \in \calP_i} f_p \geq (1 - 1/m)(1 + 1/(m-1)) = 1$.
	If $\sum_{p \in \calP_i} f_p > 1$ for some $i$, 
	we can uniformly reduce all $p \in \calP_i$ until $\sum_{p \in \calP_i} f_p = 1$.  
	The argument for \hclp~is analogous.  
\end{proof}

\subsection{Hop-Constrained Congestion-Based Embeddings}

The results in this paper use the congestion-based hop-constrained embeddings
developed in \cite{hop_congestion21}. Given a graph $G = (V, E)$,
a \emph{partial tree embedding} $(T, M)$ is a rooted tree $T$ such that
$V(T) \subseteq V$, along with a mapping $M: E(T) \to 2^{E}$
that maps each edge $uv \in E(T)$ to a path $M_{uv}$ from $u$ to $v$ in $G$.
We let $P_T(u, v) \subseteq E(T)$ denote the unique $u$-$v$ path in $T$.
We extend the definition of $M_{uv}$ for $uv \notin E(T)$ to be
the concatenation of $M_{e_1}, \dots, M_{e_t}$,
where $e_1, \dots, e_t = P_T(u,v)$.
For $e \in E(G)$, let $\flow(M_{uv}, e)$ denote the number of times
$M_{uv}$ traverses $e$.
We use the following lemma on tree distributions in \cite{hop_congestion21}.

\begin{lemma}[\cite{hop_congestion21}]
\label{lem:d1_router}
	For every complete capacitated graph $G = (V, E, x)$
	and every $\eps \in (0, 1/3)$, there exists a distribution $\calT$
	over partial tree embeddings such that:
	\begin{enumerate}
		\item \emph{(Exclusion probability)}
		For each $v \in V$, $\Pr_{(T, M) \sim \calT}[v \in V(T)] \geq 1-\eps$,
		\item \emph{(Hop dilation)}
		For each $(T, M) \sim \calT$, $\forall u, v \in V(T)$,
		$\hop(M_{uv}) \leq O\left(\frac{\log^3 n}{\eps}\right)$,
		\item \emph{(Expected congestion)} For each $e \in E$ with capacity $x_e$,
		\\ $\E\left[\sum_{u,v \in V(T)} \flow(M_{uv}, e) \cdot x_{uv}\right]
		\leq O\left(\log n \cdot \log{\frac {\log n}\eps}\right) \cdot x_e$.
	\end{enumerate}
	Furthermore, if the ratio $\frac{\max_{e \in E} x(e)}{\min_{e \in E} x(e)}$
	is at most $\poly(n)$,
	the height of all trees in the support of $\calT$ is bounded by $O(\log n)$.
	Following \cite{hop_congestion21}, we refer to these distributions
	as $\droute$-routers.
\end{lemma}

For a partial tree embedding $(T, M)$ and an edge $e' \in E(T)$,
we let $A_{e'}, B_{e'} \subseteq V(T)$ denote the two components of $T \setminus e'$.
Let $E(A_{e'}, B_{e'})$ denote the set of edges of $G$ with one endpoint in $A_{e'}$ and
the other in $B_{e'}$; these are the edges on the cut $\delta_{G[V(T)]}(A_{e'})$.
We let $y(e') = \sum_{e \in E(A_{e'}, B_{e'})} x_{e}$; we call this the 
\emph{capacity} of a tree edge $e'$. For ease of notation, 
we refer to the \emph{inverse} of the mapping function $M$ 
as $M^{-1}(e) = \{e' \in E(T): e \in M_{e'}\}$.

\begin{lemma}
\label{claim:edge_load}
	Let $e \in E(G)$. For a randomly sampled tree $(T, M)$ from a $\droute$-router, \\
	$\E\left[\sum_{e' \in M^{-1}(e)} y(e')\right]
	\leq O(\log n \log \frac{\log n}\eps) \cdot x_e$.
\end{lemma}
\begin{proof}
	Fix $e \in E$;
	$\sum_{e' \in M^{-1}(e)} y(e')
	= \sum_{e' \in M^{-1}(e)} \sum_{uv \in E(A_{e'}, B_{e'})} x_{uv}$.
	If $uv \in E(A_{e'}, B_{e'})$ for some $e' \in E(T)$, then 
	$e'$ must be on $P_T(u,v)$.
	Therefore, $x_{uv}$ is counted once for every edge $e' \in P_{T(u,v)}$ 
	such that $e \in M_{e'}$. By definition, this is $\flow(M_{uv}, e)$.
	By the congestion guarantee on the distribution of partial tree embeddings,
	$\E\left[\sum_{u, v \in E(T)} \flow(M_{uv}, e) \cdot x_{uv}\right]
		\leq O\left(\log n \cdot \log \frac{\log n} \eps\right) \cdot x_e.$
\end{proof}

\begin{remark} We will use $\droute$-routers with $x_e$ values from \hclp
~as capacities. We can assume the height of the trees in the support of $\calT$
are at most $O(\log n)$, since $x_e$ values
are between $\frac 1 {m^2}$ and 1 by Lemma \ref{lem:lp_assumption1}. Second,
we will use Lemma~\ref{lem:d1_router} on graphs that are not
complete. This is without loss of generality; we can set $x_e = \frac 1 {m^2}$
for all $e \notin E$, so the total amount of flow carried on edges not in $E$ is
at most $\frac 1 {m}$. It is easy to verify that this does not affect the
results in this paper.
\end{remark}

\subsection{Group Connectivity and Oblivious Rounding}

Set Connectivity was first studied in the single-source setting, known
as Group Steiner Tree (GST), which has an
approximation algorithm via tree embeddings \cite{GargKR98}
(see Appendix \ref{sec:related_work}).
Chekuri et al. extended the work on GST to 
Set Connectivity, obtaining a polylogarithmic approximation ratio and
integrality gap \cite{cegs11}. Motivated by higher connectivity variants,
Chalermsook, Grandoni, and Laekhanukit~\cite{CGL15} describe a
rounding algorithm based on \racke's seminal congestion-based tree embeddings
\cite{Racke08} for Set Connectivity that is oblivious to the set pairs.
We use the following lemma summarizing their result,
and refer to the algorithm as $\treeround(T, x, f)$.

\begin{lemma}[\cite{CGL15,cllz22}]
\label{lem:setconnectivity-tree-rounding}
	Consider an instance of Set Connectivity on an $n$-node tree $T=(V,E)$
	with height $h$ and let $x: E \rightarrow [0,1]$. Suppose $A, B
	\subseteq V$ are disjoint sets and suppose $K \subseteq E$ such that
	$x$ restricted to $K$ supports a flow of $f \le 1$ between $A$ and
	$B$. There is a randomized algorithm that is oblivious to $A, B, K$
	(hence depends only on $x$ and value $f$) that outputs a subset $E' \subseteq E$
	such that (i) The probability that $E' \cap K$ connects $A$ to $B$ is
	at least a fixed constant $\phi$ and (ii) For any edge $e \in E$, the
	probability that $e \in E'$ is $\min\{1,O(\frac{1}{f} h \log^2 n) x(e)\}$.
\end{lemma}
\section{LP-Based Reduction from Buy-at-Bulk to Hop-Constraints}
\label{sec:bab}

In this section, we consider Buy-at-Bulk network design. We will prove the 
following lemma:

\begin{lemma}
\label{lem:bab_reduction}
  Suppose we are given an $(\alpha, \beta)$-approximation for 
  HCSF on a graph with $r \cdot \poly(n)$ vertices with respect to 
  its optimal fractional solution.
  Then there exists an $O(\log D(\alpha \log n + \beta))$-approximation
  for multicommodity Buy-at-Bulk with respect to $\opt_{LP}$.
  This reduction can be extended to the Set Connectivity setting.
\end{lemma}

Note that Theorem \ref{thm:bab_main} follows directly as a corollary of 
Lemma \ref{lem:bab_reduction} and Theorem \ref{thm:hc_main}.
We will discuss the reduction in the set connectivity setting, as this 
generalizes multicommodity.
Let $G = (V, E)$ be a graph with edge costs 
$c: E \to \R_{\geq 0}$, edge lengths $\ell: E \to \R_{\geq 0}$, 
and demand pairs $\{S_i, T_i\}_{i \in [r]}$ 
with demands $\delta(i)$. Recall that our goal is to find $F \subseteq E$ 
minimizing $c(F) + \sum_{i \in [r]} \delta(i) \ell(p_i)$, where $p_i$ 
is the shortest $S_i$-$T_i$ path in $F$.

The reduction consists of two main components. First, we show that we can 
assume  $\delta(i) = 1$ for all $i \in [r]$, $\ell(e) = 1$ 
for all $e \in E$ with an $O(\log D)$ loss in the approximation ratio. 
The idea is to partition
the set of terminal pairs based on their demands and solve each 
subinstance separately, and also subdivide each edge into paths 
of length $\ell(e)$, where each new edge has length 1. 
There are some technical difficulties in handling settings where 
the maximum edge length is large, as this reduction could introduce 
exponentially many variables. This can be handled via scaling tricks 
on the fractional solution. The proof of these assumptions 
is given in Section \ref{sec:bab_assumptions}.

The second component under these assumptions is given by Algorithm 
\ref{algo:reduction}. Note that $\ell(e) = 1$ implies that 
$\ell(p) = \hop(p)$ for all paths $p$. Therefore, we can use 
the LP to ``guess'' the optimal hop-requirement for each demand pair. 
We then partition demand pairs based on this hop-constraint and 
solve each subinstance separately. 

\begin{algorithm}[H]
\caption{Reducing BaB to Hop-Constrained Network Design}
\label{algo:reduction}
  \begin{algorithmic}
    \State $F \gets \emptyset$
    \State $(x, f) \gets$ feasible solution to \bablp
    \State For $i \in [r]$, $h_i \gets \sum_{p \in \calP_i} \ell(p)f_p$
    \For{$j = 1, \dots, \log n + 1$}
      \State $T_j \gets \{i: h_i \in [2^{j-1}, 2^j)\}$
      \State $F_j \gets (\alpha, \beta)$-approx for HCSC with 
      terminal pairs $T_j$, hop constraint $2^{j+1}$
      \State $F \gets F \cup F_j$
    \EndFor \\
    \Return F
  \end{algorithmic}
\end{algorithm}

\begin{lemma}
\label{lem:reduction_routingcost}
  The total routing cost given by Algorithm \ref{algo:reduction} 
  is at most $4\beta \cdot \opt_{LP}$.
\end{lemma}
\begin{proof}
  Fix $i \in [r]$. Since $\ell(p) = \hop(p)$, $\ell(p) \leq n$ 
  for all $p \in \calP_i$. Thus $h_i \leq n \sum_{p \in \calP_i} f_p = n$ 
  and $h_i \geq \sum_{p \in \calP_i} f_p = 1$, 
  so $h_i \in [2^0, 2^{\log n + 1})$.  
  Let $j$ be such that $i \in T_j$. 
  Since $F_j$ is an $(\alpha, \beta)$-approximation to $2^{j+1}$-HCSF, 
  $F_j$ contains an $S_i$-$T_i$ path of length at most 
  $\beta 2^{j+1}$. By construction of $T_j$, $h_i \geq 2^{j-1}$, so 
  $\beta 2^{j+1} \leq 4 \beta h_i$. Thus the total routing cost over 
  all pairs is at most $4\beta \sum_{i \in [r]} h_i 
  = 4\beta \sum_{i \in [r]} \sum_{p \in \calP_i} \ell(p) f_p 
  \leq 4\beta \cdot \opt_{LP}$.
\end{proof}

\begin{lemma}
\label{lem:reduction_fixedcost}
  The total fixed cost $\sum_{e \in F} c(e)$ given by Algorithm \ref{algo:reduction} 
  is at most $O(\alpha \log n) \cdot \opt_{LP}$.
\end{lemma}
\begin{proof}
  Fix an iteration $j$. Let $x_e' = \max\{2x_e, 1\}$, and $f_p' = \max\{2f_p, 1\}$. 
  We claim that $(x', f')$ is a valid solution to \hclp 
  ~with terminal pairs $T_j$ and hop constraint $h = 2^{j+1}$. 
  First, we show that for all $i \in T_j$, $\sum_{p \in \calP_i^h} f_p' = 1$. 
  Suppose for the sake of contradiction that $\sum_{p \in \calP_i^h} f_p' < 1$. 
  Then $\sum_{p \in \calP_i^h} f_p < \frac 12$. Since $f_p$ satisfies the 
  Buy-at-Bulk LP constraints, it must be the case that 
  $\sum_{p \in \calP_i, \ell(p) > h} f_p > \frac 12$, implying that 
  $\sum_{p \in \calP_i} \ell(p)f(p) > \frac h 2 = 2^j$. 
  This contradicts the assumption that $i \in T_j$. Thus 
  $\sum_{p \in \calP_i^h} f_p' \geq 1$. Note that if this sum is greater than 1, 
  we can reduce the $f_p'$ values until it is exactly 1 without violating 
  any constraints. 

  Next, we show that for all $i \in T_r$, $e \in E$, 
  $\sum_{p \in \calP_i^h, e \in p} f_p' \leq x_e'$. 
  First, suppose $x_e' = 1$. 
  Then, for all $i \in T_j$, $\sum_{p \in \calP_i^h, e \in p} f_p' 
  \leq \sum_{p \in \calP_i^h} f_p' = 1$ (by the earlier paragraph), 
  so the constraint is satisfied. 
  Instead, suppose $x_e' = 2x_e$. In this case,
  $\sum_{p \in \calP_i^h, e \in p} f_p' 
  \leq 2 \sum_{p \in \calP_i, e \in p} f_p \leq 2 x_e = x_e'$, 
  once again satisfying the constraint.
  Let $\opt_{LP}'$ denote the value of the optimal LP solution 
  with respect to \hclp~ for the HCSC instance given in iteration $j$.
  Then $\opt_{LP}' \leq \sum_{e \in E} c(e)x_e' \leq 2\opt_{LP}$. 
  By construction, $c(F_j) \leq \alpha \cdot \opt_{LP}'$. 
  Summing over all $j$, $c(F) \leq O(\alpha \log n) \opt_{LP}$.
\end{proof}

\begin{proof}[Proof of Theorem~\ref{thm:bab_main}]
  By Claims \ref{claim:bab_demand_assumption} and \ref{claim:bab_length_assumption},
  we can assume $\delta(i) = 1~\forall i \in [r], \ell(e) = 1~\forall e \in E$
  with an $O(\log D)$ loss in the approximation factor. 
  With these assumptions, by Lemmas \ref{lem:reduction_routingcost} 
  and \ref{lem:reduction_fixedcost}, an $(\alpha, \beta)$-approximation 
  for HCSC with respect to the LP
  gives an $O(\alpha \log n + \beta)$-approximation 
  for Set Connectivity Buy-at-Bulk with respect to $\opt_{LP}$.  
  It is easy to verify that the same reduction holds from
  HCSF to MC-BaB.
\end{proof}

\begin{remark}
\label{rem:bab_logD}
  We believe the $\log D$ factor is not needed in the upper bound on the 
  integrality gap. We are able to remove this factor in 
  some special cases, see Claim \ref{claim:removing_logD_factor}.
  In general, this has been difficult to avoid in several previous approaches 
  for nonuniform buy-at-bulk (see Section \ref{sec:related_work}). 
  One technical reason is that we rely on hop-constrained tree embeddings 
  in a black-box fashion and that work relies on the assumption that $\ell(e) = 1$ 
  for each edge. Handling non-uniform lengths that may not be polynomially 
  bounded in $n$ appears to be a technical challenge. 
\end{remark}

\begin{remark}
\label{rem:bab_reduction_direction}
  It is not difficult to see that a similar reduction can be made
  from $h$-HCSF to MC-BaB. One could use an iterative approach;
  set all $\ell(e) = 1$ and solve MC-BaB. Buy $h$-hop-bounded paths 
  connecting terminal pairs, double $\ell(e)$, and recurse on remaining 
  pairs. Thus new techniques directly 
  addressing buy-at-bulk could imply good approximations for 
  hop-constrained counterparts.
\end{remark}

\subsection{Length and Demand Assumptions}
\label{sec:bab_assumptions}

\begin{lemma}
\label{lem:lp_assumption2}
  Let $(x, f)$ be an optimal solution to \bablp. Losing a constant factor in
  the approximation ratio, we assume for all $e \in E$, there exists $i \in [r]$ such that
  $\sum_{p \in \calP_i, e \in p} f_p = x_e$. We do this without violating the
  assumption that $x_e \geq 1/m^2$ for all $e \in E$.
\end{lemma}
\begin{proof}
  Let $(x, f)$ be an optimal solution to \bablp. 
  Recall that in order to assume that 
  $x_e \geq \frac 1 {m^2}$ for all $e \in E$, we set $x_e = 0$ for all 
  $x_e < \frac 1 {m^2}$ and account for this by scaling all remaining edges 
  up by $1+ 1/(m-1)$. Suppose there is some 
  $x_e > \max_{i \in [r]} \sum_{p \in \calP_i, e \in p} f_p$. 
  Then, we can decrease all 
  $x_e$ to $\max_{i \in [r]} \sum_{p \in \calP_i, e \in p} f_p$ without 
  violating any constraints, and then apply the above transformation 
  again to ensure that $x_e \geq \frac 1 {m^2}$. We repeat alternating between 
  the two transformations until both requirements are satisfied. 
  At each iteration, 
  at least one $x_e$ is set to 0. 
  Thus this process can be repeated at most $m-1$ times, 
  so the total increase to the approximation ratio is at most 
  $(1 + 1/(m-1))^m \leq e$.
\end{proof}

\begin{claim}
\label{claim:bab_demand_assumption}
  We can assume $\delta(i) = 1$ for all $i \in [r]$ with an 
  $O(\log D)$ factor loss in the approximation ratio.
\end{claim}
\begin{proof}
  Recall that we assume $\min_{i \in [r]} \delta(i) = 1$,
  so $\delta(i) \leq D$ for all $i \in [r]$.
  We round each demand up to its closest power of 2, 
  losing a factor of at most 2 in the approximation ratio, so each 
  $\delta(i) = 2^j$ for $j \in \{0, \dots, \log D + 1\}$. 
  For each $j$, we solve the Buy-at-Bulk problem restricted 
  to pairs with demand $2^j$. In each sub-instance, 
  we assume $\delta(i) = 1$ by dividing all $c(e)$ by $2^j$. 
  By taking the union of all sub-instances, 
  an $\alpha$-approximation when $\delta(i) = 1$ gives an 
  $O(\alpha \log D)$-approximation in general. 
\end{proof}

As discussed in Remark \ref{rem:bab_logD}, we can avoid the $O(\log D)$ factor
loss in the approximation ratio in the case that $D$ or $L$ are polynomially 
bounded, as shown in the following claim.
\begin{claim}
\label{claim:removing_logD_factor}
  If $D = r \cdot \poly(n)$ or $L = \poly(n)$, we can assume $\delta(i) = 1$ for all 
  $i \in [r]$ with only a constant factor loss in the approximation 
  ratio by increasing the number of demand pairs to $r^2 \cdot \poly(n)$.
\end{claim}
\begin{proof}
  Suppose $D = r \cdot \poly(n)$. For each $i \in [r]$, 
  replace $i$ with $\delta(i)$ demand pairs from $s_i$ to $t_i$ of 
  demand 1 each. It is clear that the total cost stays the same 
  and the total number of demand pairs is at most $r \cdot D \leq r^2\poly(n)$.

  Suppose $L = \poly(n)$. Let $\ell_{\max}$, $\ell_{\min}$ 
  denote the maximum and minimum $\ell(e)$ values respectively.
  For all $i \in [r]$ with $\delta(i) < D/(nrL)$, we 
  increase $\delta(i)$ to 
  $D/(nrL)$. For any fractional solution $(x, f)$, the increase to 
  the routing cost is at most $\sum_{i \in [r], \delta(i) < D/(nrL)} 
  (D/(nrL) - \delta(i)) \sum_{p \in \calP_i} \ell(p)f_p$. 
  Note that $\ell(p) \leq n\ell_{\max}$, since any 
  simple path has at most $n-1$ edges. Furthermore, for all $i\in [r]$, 
  $\sum_{p \in \calP_i} f_p = 1$. Thus the increase to the routing cost 
  is at most $r n (D/(nrL)) \ell_{max} 
  \leq D \ell_{\min}$. 
  Let $i_{\max}$ be the index such that $\delta(i_{\max}) = D$.  
  The total routing cost must be at least 
  $D\sum_{p \in \calP_{i_{\max}}} \ell(p)f_p \geq D \ell_{\min}$. 
  Thus increasing the demands of small demand pairs to $D/(nrL)$ at most 
  doubles the total routing cost. We scale
  to set the minimum demand to 1, reducing to the case where 
  $D = r \cdot \poly(n)$.
\end{proof}

\begin{claim}
\label{claim:bab_length_assumption}
  Suppose $\delta(i) = 1$ for all demand pairs $i \in [r]$. Then we can assume 
  $\ell(e) = 1$ for all $e \in E$ by increasing the size of the graph 
  to $r \cdot \poly(n)$.
\end{claim}
\begin{proof}
  We begin by rounding up each $\ell(e)$ to its closest power of 2, 
  losing a factor of at most 2 in the approximation ratio. 
  The high level idea is to subdivide all 
  $\ell(e)$ into $\ell(e)$ copies of length 1 each. However, the problem 
  is that if $\ell(e)$ is large for some $e \in E$, this approach 
  may introduce too many new edges and vertices into the graph. To overcome 
  this issue, we essentially ignore all edges whose lengths are significantly 
  smaller than the maximum edge length. 

  Let $\ell_{\max} = \max_{e \in E} \ell(e)$. 
  Let $\ell^* = \max(\frac{\ell_{\max}}{rnm^2}, 1)$ rounded up to its closest 
  power of 2. We construct a new graph $G' = (V', E')$ with edge costs and lengths 
  $c'$ and $\ell'$ as follows: for all $e \in E$, if $\ell(e) \leq \ell^*$, then 
  we keep $e \in E'$ as is. Else, we replace $e$ with a path $p_e$ 
  consisting of $\frac{\ell(e)}{\ell^*}$ edges of cost 
  $\frac{c(e) \ell^*}{\ell(e)}$ each. We set $\ell'(e') = \ell^*$ for
  all $e' \in E$. For ease of notation, we let $p_e$ denote the path consisting 
  of a single edge $e$ for all edges with $\ell(e) \leq \ell^*$. 
  Notice that $c'(p_e) = c(e)$ for all $e$; this is clear if $\ell(e) \leq \ell^*$,
  and if $\ell(e) > \ell^*$, 
  then $c'(p_e) = \sum_{e' \in p_e} c'(e') = \frac{\ell(e)}{\ell^*} 
  \frac{c(e) \ell^*}{\ell(e)} = c(e)$. 
  Also notice that if $\ell(e) \geq \ell^*$, then $\ell'(p_e) = 
  \ell^*\frac{\ell(e)}{\ell^*} = \ell(e)$. 

  Let $(x, f)$ be an optimal LP solution on $G$. We start by showing that the 
  optimal LP solution on $G'$ is has cost at most thrice that of $(x, f)$. 
  We construct $(x', f')$ as follows: For each $e \in E$, 
  for all $e' \in p_e$, set $x_{e'}' = x_e$. For each $i \in [r]$, 
  path $p \in \calP_i$, it is easy to see that the same path exists in $G'$ 
  where each edge $e \in p$ is replaced by the path $p_e$; this is because all 
  new paths $p_e$ are disjoint. We will abuse notation and let $p$ denote 
  the path in $G$ and the corresponding path in $G'$. For all paths $p$, 
  we let $f_p' = f_p$. The new total fixed cost is 
  $\sum_{e' \in E'} c'(e')x_{e'}' = \sum_{e \in E} c(p_e) x_e
  = \sum_{e \in E} c(e) x_e$.  
  We use the fact that $\delta(i) = 1$ for all $i \in [r]$ to get 
  a routing cost of 
  $$\sum_{i \in [r]} \sum_{p \in \calP_i} f_p' \sum_{e \in p} \ell'(p_e)
  = \sum_{i \in [r]} \sum_{p \in \calP_i} f_p \sum_{e \in p} \max(\ell^*, \ell(e)) 
  \leq \sum_{i \in [r]} \sum_{p \in \calP_i} f_p \ell(p) 
  + \ell^* \sum_{i \in [r]} \sum_{p \in \calP_i} f_p|p|$$
  
  If $\ell^* = 1$ then $\ell'(p_e) = \ell(e)$ for all $e \in E$, so we assume 
  instead that $\ell_{\max} \geq rnm^2$. Since we round $\ell^*$ up to 
  its closest power of 2, $\ell^* \leq \frac{2\ell_{\max}}{rnm^2}$.
  By the LP constraints, $\sum_{p \in \calP_i} f_p = 1$ for all $i \in [r]$, 
  and since $\calP_i$ contains only simple paths in $G$, $|p| < n$. 
  Therefore, $\ell^* \sum_{i \in [r]} \sum_{p \in \calP_i} f_p|p| 
  \leq rn\ell^* \leq 2\ell_{max}/m^2$. We need to show that this is not too much 
  larger than the optimal routing cost in $G$. By Lemma \ref{lem:lp_assumption2},
  for $e^*$ such that 
  $\ell(e^*) = \ell_{\max}$, there is some $i \in [r]$ such that 
  $\sum_{p \in \calP_i, e^* \in p} f_p = x_{e^*} \geq \frac 1 {m^2}$. 
  Thus $\sum_{i \in [r]} \sum_{p \in \calP_i} f_p \ell(p) 
  \geq \sum_{i \in [r]} \sum_{p \in \calP_i, e^* \in p} f_p \ell_{\max} 
  \geq \ell_{\max}/m^2$. 
  Therefore, the total routing cost of $f'$ is at most 3 times the 
  routing cost of $f$. It is simple to verify that 
  $(x', f')$ satisfy the LP constraints for $G'$. 

  The other direction is simple: given an optimal LP solution 
  $(x', f')$ on $G'$, we can define an LP solution $(x, f)$ on $G$ 
  as follows: For each edge $e \in E$, we let 
  $x_e = \min_{e' \in p_e} x_{e'}'$. 
  For each path $p \in \calP_i$, we let $f_p = f_{p}'$. It is easy 
  to see that the fixed costs and routing costs of $(x, f)$ are at 
  most those of $(x', f')$ , and that all LP constraints are satisfied. 

  Therefore, we can assume that $\ell(e) = \ell^*$ for all $e$ with only 
  a constant factor loss in the approximation ratio. 
  To complete the proof, we divide all $\ell(e)$ and $c(e)$ by $\ell^*$, 
  since the minimization problem does not change when scaling all values by 
  a constant. This transformation adds at most 
  $r \cdot \poly(n)$ new vertices and edges. 
\end{proof}
\section{Approximation Algorithms for Hop-Constrained 
Network Design Problems via Oblivious Routing}
\label{sec:hc}

In this section we prove Theorem \ref{thm:hc_main}. We give
an overview of the algorithm in Section \ref{sec:hc_algo}, which we apply to 
Hop-Constrained Steiner Forest in Section \ref{sec:hc_sf} and 
Hop-Constrained Set Connectivity in Section \ref{sec:hc_sc}.

\subsection{Algorithm Overview}
\label{sec:hc_algo}

Let $G = (V, E)$ with edge costs $c(e)$ and hop-constraint $h$ 
denote the given input. Recall that our goal is to find a subgraph 
$F \subseteq E$ minimizing $c(F)$ in which each demand pair is connected 
with a path of hop-length at most $h$. We employ a tree based 
rounding scheme described in Algorithm \ref{algo:hop_general}. 
The algorithm solves \hclp, 
obtains a fractional solution $(x, f)$, 
and considers a $\droute$-router with capacities $x_e$ and exclusion 
probability $\eps = O(1/h)$. It repeatedly sample partial tree embeddings 
and buys tree paths through some randomized process to 
be described separately for each problem.  

\begin{algorithm}[H]
\caption{$(\alpha, \beta)$-approximation for Hop Constrained Network Design}
\label{algo:hop_general}
  \begin{algorithmic}
    \State $F \gets \emptyset$
    \State $\calT \gets$ $\droute$-router on $G$ with capacities $x_e$, 
    exclusion prob $\eps = \frac 1 {4(h+1)}$ 
    \For{$j = 1, \dots, \tau$}
      \State Sample $(T, M) \sim \calT$
      \State $F_j' \subseteq E(T) \gets$ 
      \text{ Some collection of tree paths}
      \State $F_j \gets \cup_{e' \in F_j'} M_{e'}$, this is the mapping of $F_j'$ to 
      the input graph $G$
      \State $F \gets F \cup F_j$
    \EndFor \\
    \Return $F$
  \end{algorithmic}
\end{algorithm}

\begin{lemma}
\label{lem:hc_cost_general}
  Suppose that in every iteration $j$, the probability that 
  $e' \in E(T)$ is added to $F_j'$ is at most $\gamma \cdot y(e')$. 
  Then the total cost of Algorithm \ref{algo:hop_general} is at most 
  $O(\gamma \tau \log^2n) \sum_{e \in E} c(e) x_e$.
\end{lemma}
\begin{proof}
  Fix an iteration $j \in [\tau]$ and let $(T, M)$ be the sampled 
  partial tree embedding. 
  For each edge $e \in E$, $e \in F_j$ only if $e \in M_{e'}$ for some 
  $e' \in F_j'$, i.e. $F_j' \cap M^{-1}(e) \neq \emptyset$. 
  Thus
  $\Pr[e \in F_j] \leq \sum_{e' \in M^{-1}(e)} \Pr[e' \in F_j'] 
  \leq \sum_{e' \in M^{-1}(e)} \gamma \cdot y(e')$. By Lemma 
  \ref{claim:edge_load}, this is at most 
  $O(\gamma \log n \cdot \log \frac{\log n} \eps) \cdot x_e$. 
  Since we set $\eps = \frac 1 {4(h+1)}$, $O(\log \frac{\log n} \eps)
  = O(\log \log n + \log h) = O(\log n)$. 
  Thus $\E[c(F_j)] = \sum_{e \in E} c(e) \Pr[e \in F_j] 
  \leq O(\gamma \log^2 n) \sum_{e \in E} c(e) x_e$. Summing over all 
  $\tau$ iterations gives the desired total cost.
\end{proof}

We show that the tree preserves $S_i$-$T_i$ 
flow with good probability. Since set demand pairs generalize
vertex demand pairs, the following lemma holds for all settings.

\begin{lemma}
\label{lem:hc_flow_preservation}
  Fix a demand pair $S_i, T_i \subseteq V$ and an iteration $j \in \tau$. 
  With probability at least
  $1/2$, $y$ supports a flow of at least $1/2$ between $S_i$ and $T_i$ on $T$.
\end{lemma}
\begin{proof}
  By construction, $y$ supports at least as much flow between $S_i$ and $T_i$ 
  on $T$ as $x$ does on $G[V(T)]$, where $G[V(T)]$ denotes the 
  induced subgraph of $G$ by $V(T)$. 
  Let \\$z_e = \sum_{p \in \calP_i^h, V(p) \subseteq V(T), e \in p} f_p$. 
  Since $x, f$ is a feasible solution to \hclp, $x_e \geq z_e$ for all $e \in E$, so we
  restrict our attention to flow supported by $z$. 
  Let $Z$ denote the total $S_i$-$T_i$ flow supported by $z$, i.e.
  $Z = \sum_{p \in \calP_i^h, V(p) \subseteq V(T)} f_p$. 
  It suffices to show that $Z \geq \frac 12$ with probability at least $\frac 1 2$. 
  
  First, note that for all 
  $p \in \calP_i^h$, $\hop(p) \leq h$, so $|V(p)| \leq h+1$. Thus by union bound, 
  $\Pr[V(p) \not \subseteq V(T)] \leq \sum_{v \in V(p)} \Pr[v \notin V(T)] 
  \leq \eps (h+1)$. 
  Therefore, $\E[Z] = \sum_{p \in \calP_i} f_p \Pr[V(p) \subseteq V(T)] 
  \geq \sum_{p \in \calP_i} f_p (1 - \eps (h+1)) = (1 - \eps (h+1)) =  \frac 34$. 
  Since $\sum_{p \in \calP_i} f_p = 1$ (constraint (2b) on \hclp),  
  $Z \in [0,1]$. Applying Markov's inequality on $1-Z$ gives 
  $\Pr\left[Z \leq \frac 1 2\right] = \Pr[1 - Z \geq \frac 1 2] 
  \leq 2 \E[1 - Z] \leq \frac 1 2$.
  Thus $Z \geq \frac 1 2$ with probability $\frac 1 2$.
\end{proof}

\subsection{Hop-Constrained Steiner Forest}
\label{sec:hc_sf}

We are given terminal pairs $s_i, t_i \subseteq V$, 
$i \in [r]$, and the goal is to ensure that $F$ contains an 
$s_i$-$t_i$ path of hop-length $h$. We follow Algorithm 
\ref{algo:hop_general} with $\tau = \log_2(2r)$. 
For each iteration $j \in [\tau]$,  
for each $i \in [r]$, if $y$ supports a flow of at least $\frac 12$ 
between $s_i$ and $t_i$, we set $F_j' \gets F_j' \cup P_T(s_i, t_i)$.

\begin{lemma}
\label{lem:hcsf_hopstretch}
  With probability at least $\frac 1 2$, $F$ contains a path of length 
  at most $O(\log^3 n) \cdot h$ for each $s_i,t_i$ pair.
\end{lemma}
\begin{proof}
  By the guarantee on $\calT$, $\hop(M_{s_it_i}) \leq O(\log^3 n)\cdot h$ 
  for all $(T, M)$ in the support of $\calT$, all 
  $s_i, t_i \in V(T)$. Therefore, it suffices to show that with 
  probability at least $\frac 1 2$, for every demand pair there exists 
  an iteration $j$ in which we include the path $M_{s_it_i}$ in $F_j$. 
  Fix $i \in [r]$. In a single iteration, 
  we buy $M_{s_it_i}$ if $y$ supports a flow of at least $\frac 1 2$ 
  between $s_i$ and $t_i$, which happens with probability $\frac 1 2$ by 
  Lemma \ref{lem:hc_flow_preservation}. Therefore, the probability that 
  an $s_i$-$t_i$ path is never bought in any iteration is at most 
  $\left(\frac 1 2\right)^\tau$. By union bound over all demand pairs, 
  the probability that there exists $i \in [r]$ that does not have 
  a short path in $F$ is at most
  $r \left(\frac 1 2\right)^\tau$. Setting $\tau = \log_2(2r)$ gives the 
  desired failure probability of $\frac 1 2$.
\end{proof}

\begin{proof}[Proof of Theorem~\ref{thm:hc_main} for $h$-HCSF]
  For any $j \in [\tau]$, 
  if $e' \in F_j'$, then there exists $i \in [r]$ such that
  $e' \in P_T(s_i, t_i)$ and $y$ supports a flow of at least $\frac 12$ 
  on $P_T(s_i, t_i)$. Therefore, $e' \in F_j'$ only if 
  $y(e') \geq \frac 1 2$, so the probability that $e' \in F_j'$ 
  is at most $1 \leq 2y(e')$. By Lemma \ref{lem:hc_cost_general} with 
  $\gamma = 2$ and $\tau = \log_2(2r)$, the expected cost of Algorithm
  \ref{algo:hop_general} for HCSF is at most 
  $O(\log r \log^2n) \sum_{e \in E} c(e) x_e$, giving an expected cost 
  factor of $\alpha = O(\log r \log^2n)$.
  By Lemma~\ref{lem:hcsf_hopstretch}, the expected hop factor is 
  $\beta = O(\log^3 n)$.
\end{proof}

\subsection{Hop-Constrained Set Connectivity}
\label{sec:hc_sc}

Let $S_i, T_i \subseteq V$, $i \in [r]$ denote the demand pairs; we want to find
$F \subseteq E$ that contains an $s_i$-$t_i$ path of hop-length $h$.
In Set Connectivity, the fractional flow between $S_i$ and $T_i$ may be 
spread out amongst many different vertex pairs $(s_i, t_i) \in S_i \times T_i$. 
Therefore, we cannot follow the same simple algorithm used for 
Hop-Constrained Steiner Forest. Instead, we use the 
Oblivious Tree Rounding framework.
In each iteration $j \in [\tau]$ of Algorithm \ref{algo:hop_general}, we
let $F_j' = \treeround(T, y, \frac 1 2)$. 
We set  $\tau = \frac 2 \phi \log{2r}$, 
where $\phi$ is the probability of success given by Lemma 
\ref{lem:setconnectivity-tree-rounding}.

\begin{lemma}
\label{lem:hcg_hopstretch}
	With probability at least $\frac 1 2$, $F$ contains a 
	path of length at most $O(\log^3 n) h$ for each $S_i,T_i$ pair.
\end{lemma}
\begin{proof}
	By the guarantee on $\calT$, any $S_i$-$T_i$ path 
	on a sampled tree $T$ has hop length at most $O(\log^3 n)\cdot h$. 
	Therefore, it suffices to show that with probability at least $\frac 1 2$, 
	for every demand pair $i \in [r]$ there exists an iteration $j \in [\tau]$ 
	in which we buy an $S_i$-$T_i$ path in the tree.
	Fix an iteration $j \in [\tau]$ and let $(T, M)$ be the sampled partial tree embedding. 
	By Lemma~\ref{lem:hc_flow_preservation}, 
	$y$ supports a flow of at least $\frac 1 2$ on $T$ between $S_i$ and $T_i$ 
	with probability at least $\frac 1 2$.
	By Lemma~\ref{lem:setconnectivity-tree-rounding}, 
	given that $y$ supports a flow of at least $\frac 1 2$ on $T$ 
	between $S_i$ and $T_i$, the Oblivious Tree Rounding algorithm will connect 
	$S_i$ and $T_i$ with probability at least $\phi$. 
	Thus the probability that iteration $j$ connects $S_i$ and $T_i$ is 
	at least $\frac \phi 2$, so none of the iterations connect 
	$S_i$ and $T_i$ with probability at most $(1 - \frac \phi 2)^\tau$. 
	By union bound 
	over all demand pairs, the algorithm fails with probability at most 
	$r(1-\frac \phi 2)^\tau \leq r e^{-\phi \tau/2}$. 
	Since $\tau = \frac 2 \phi \log{2r}$, we get the 
	desired failure probability of $\frac 1 2$.
\end{proof}

\begin{proof}[Proof of Theorem \ref{thm:hc_main} for $h$-HCSC]
	By Lemma \ref{lem:setconnectivity-tree-rounding}, for any $e' \in E(T)$
	and any $j \in [\tau]$, the probability $e' \in F_j'$ is at most 
	$O(2 \cdot \textnormal{height}(T) \log^2n) y(e')$. Since we can assume 
	the height of $T$ is $O(\log n)$ (see Section~\ref{sec:prelim}), 
	$\Pr[e' \in F_j'] \leq O(\log^3n) y(e')$.
	By Lemma \ref{lem:hc_cost_general} with 
	$\gamma = O(\log^3 n), \tau = O(\log r)$, the expected cost of Algorithm
	\ref{algo:hop_general} for HCSC is at most 
	$O(\log r \log^5n) \sum_{e \in E} c(e) x_e$, giving an
	expected cost factor of $\alpha = O(\log r \log^5n)$.
	By Lemma~\ref{lem:hcg_hopstretch}, the expected hop factor is  
	$\beta = O(\log^3 n)$.
\end{proof}

\section{Fault-Tolerant Hop-Constrained Network Design}
\label{sec:hc_fault}

In this section, we will provide polylogarithmic bicriteria approximations 
for $(h,k)$-Fault Tolerant Hop-Constrained Network Design. We are given 
$G = (V, E)$ with edge costs $c(e)$ and terminal pairs $s_i, t_i$ with 
connectivity requirements $k_i$. Recall that a solution $F \subseteq E$ is 
feasible if for every $i \in [r]$, for every subset $Q$ of less than 
$k_i$ edges, $F$ contains need an $s_i$-$t_i$ path of hop-length at most $h$
avoiding $Q$. We modify the LP relaxation \hclp~ for the fault-tolerant setting
as follows:

\mypara{LP Relaxation:}
We define variables $f_p^Q$ $\forall p \in \calP_i^h, 
\forall Q \subseteq E, |Q| < k$, where $f_p^Q$ is an indicator for 
whether the path $p$ carries $s_i$-$t_i$ flow that avoids edges in $Q$. 
We replace the first constraint in \hclp~ with 
\[\sum_{p \in \calP_i^h, p \cap Q = \emptyset} f_p^Q = 1 \quad \forall i \in [r], 
\forall Q \subseteq E \text{ s.t. } |Q| < k_i.\] 
We replace the second constraint in \hclp~ with 
\[\sum_{p \in \calP_i^h, e \in p} f_p^Q \leq x_e \quad \forall i \in [r], 
\forall e \in E, \forall Q \subseteq E \text{ s.t. } |Q| < k_i.\] 
To solve this LP, for each $Q \subseteq E$, consider the LP 
restricted to variables $f^Q$ and constraints corresponding to $Q$. This 
is exactly \hclp. Thus 
the polytope for Fault-HCND is the intersection of $\binom m k = 
n^{O(k)}$ polytopes which we can separate over by Remark 
\ref{rem:lp_separation}.

In Section \ref{sec:hc_fault_ss}, 
we focus on the single-source setting and prove Theorem \ref{thm:hcfault_main}. 
In Section \ref{sec:junction}, we prove Theorem \ref{thm:junction_main} via 
a reduction from multicommodity to single-source when $k = 2$. 

\subsection{Single-Source Fault Tolerant HCND}
\label{sec:hc_fault_ss}

We are given a source terminal $s$ and sink terminals 
$t_1, \dots, t_r$ with requirements $k_i$ each. 
We will assume that $k_i = k$ for all $i$; this can be done with an additional 
factor of $k$ in the cost approx factor by partitioning the set of terminals 
based on $k_i$ and considering each set independently.

\subsubsection{Algorithm}

We employ the \emph{augmentation framework}. 
For $\ell \in [k]$, we say a partial solution $H_\ell$ satisfies 
$(h_\ell, \ell)$-hop-connectivity if $\forall i \in [r]$, $s$ and $t_i$ are 
$(h_\ell, \ell)$-hop-connected
(where $h_\ell \geq h$). 
Our goal is to \emph{augment} $H$ to satisfy $(h_{\ell+1}, \ell+1)$-hop
connectivity, for some $h_{\ell+1}$ proportional to $h_{\ell}$.
We modify the tree-rounding scheme given in Section \ref{sec:hc} to solve 
the augmentation problem.
As described in 
Section \ref{sec:intro}, one difficulty in using congestion-based 
tree embeddings in higher connectivity is bounding the amount of
flow sent along any edge. To 
circumvent this issue, we follow an approach introduced by 
\cite{cllz22} to scale down $x_e$ values 
of existing edges and ``large'' edges. 
Let $\sigma = O\left(\log n \cdot \log{\frac {\log n}\eps}\right)$ 
denote the congestion parameter of $\droute$-routers, i.e. 
$\E\left[\sum_{u,v \in V(T)} \flow(M_{uv}, e) \cdot x_{uv}\right]
\leq \sigma \cdot x_e$.

\begin{algorithm}[H]
\caption{Augmentation Algorithm for Fault Tolerant Hop Constrained 
	Network Design}
\label{algo:fault_hop}
	\begin{algorithmic}
		\State $(x, f) \gets$ fractional solution to LP relaxation. 
		\State $\LG \gets \{e : x_e > \frac 1 {4\ell\sigma}\}$
		\State $\SM \gets \{e: x_e \leq \frac 1 {4\ell\sigma}\}$
		\State $\tilde x_e \gets
		\begin{cases}
			x_e & e \in \SM \\
			\frac 1 {4\ell\sigma} &e \in H_\ell \cup \LG 
		\end{cases}$
		\State $\calT \gets$ $\droute$-router on $G$ with capacities $\tilde x_e$, 
		exclusion prob $\eps = \frac 1 {4(h+1)}$ 
		\For{$j = 1, \dots, \tau = O(\ell \log n)$}
		  \State Sample $(T, M) \sim \calT$, with capacities $y$ as defined in 
		  Section \ref{sec:prelim}
		  \State $F_j' \subseteq E(T) \gets \treeround(T, y, \frac 1 4)$ 
		  \State $F \gets F \cup \left(\bigcup_{e \in F_j'} M(e')\right)$
		\EndFor \\
		\Return $F \cup \LG$
	\end{algorithmic}
\end{algorithm}

We make a few assumptions. First, we assume $H_\ell$ is 
minimal; i.e. for all $e \in H_\ell$, $H_\ell \setminus \{e\}$ does not 
satisfy $(h_{\ell}, \ell)$-hop-connectivity. This can be accomplished by 
removing each $e \in H_{\ell}$ and checking if $H_{\ell} \setminus \{e\}$ 
satisfies $(h_{\ell}, \ell)$-hop-connectivity. Second, we assume that if either 
endpoint of $e$ has hop-distance greater than $h$ to $s$, then $x_e = 0$. 
This assumption is true for all optimal fractional solutions, since 
$x_e$ only needs to be nonzero if $e \in p$ for some $p \in \calP_i^h$. 
Finally, since we buy all edges in \LG~at the end of Algorithm 
\ref{algo:fault_hop} and for ease of notation, we write $H$ to denote  
$H_{\ell} \cup \LG$.

\subsubsection{Overview of Analysis}

We let $\hop_G(u,v) = \dist_G(u,v)$ denote the minimum hop-length
between two nodes in a graph $G$ (we will drop $G$ if it is clear from context). 
The diameter of $S \subseteq V$ is $\diam(S) = \max_{u, v \in S} \hop(u,v)$. 
When we discuss the diameter of the graph, we are referring to the set 
of all non-isolated vertices.

\begin{claim}
\label{claim:H-diameter}
	$H$ has diameter at most $2h_\ell$.
\end{claim}
\begin{proof}
	We will show that $\hop(s, u) \leq h_\ell$ for all $u \in H$. First, 
	suppose $u \in H$ before adding edges in $\LG$.
	Recall that $H$ satisfies $(h_\ell, \ell)$-hop connectivity.
	If $u$ is on any $s$-$t_i$ path of length at most $h_\ell$, then 
	$\hop(s, u) \leq h_\ell$. Else, 
	removing $u$ and all its incident edges would not change the feasibility 
	of $H$, contradicting minimality of $H$. Thus $\hop(s, u) \leq h_\ell$.
	Suppose instead that $u$ was added to $H$ due to the addition of 
	some edge in $\LG$. Then, there exists some edge $e$ incident to 
	$u$ such that $x_e > 0$. By assumption, $\hop(s, u) \leq h\leq h_\ell$. 
	Thus any two nodes on $H$ have a path of hop length at most $2h_\ell$ 
	by going through $s$, so $H$ has diameter at most $2h_\ell$.
\end{proof}

We say $Q \subseteq E$ is a \emph{violating edge set} if $|Q| \leq \ell$ and 
$H \setminus Q$ does not contain a path of length $h_{\ell}$ from 
$s$ to $t_i$. Since $H$ satisfies $(h_{\ell}, \ell)$-hop connectivity, 
$|Q| = \ell$ for any violating edge set.
To prove the correctness of Algorithm \ref{algo:fault_hop}, 
we need to show that $H \cup F \setminus Q$ contains a short $s$-$t_i$ path
for every violating edge set $Q$. 
For the remainder of the analysis, we fix a violating edge set $Q$.

\begin{claim}
\label{claim:H-components}
	$H \setminus Q$ has at most two connected components. Furthermore, each 
	component of $H \setminus Q$ must contain a terminal.
\end{claim}
\begin{proof}
	Suppose there exists some component $C_j$ of $H \setminus Q$ that does not 
	contain a terminal. Since $H$ is connected (a corollary of Claim 
	\ref{claim:H-diameter}), there must be some $e \in Q$ with exactly one 
	endpoint in $C_j$. Let $Q' = Q \setminus \{e\}$. Since $H$ satisfies 
	$(h_\ell, \ell)$-hop-connectivity, every terminal contains a path of hop 
	length at most $h_\ell$ to $s$ in $H \setminus Q'$. Since $C_j$ contains no 
	terminals and $e$ is the only edge connecting $C_j$ to the rest of the graph, 
	any $s$-$t_i$ path must traverse $e$ an even number of times, which means any 
	simple $s$-$t_i$ path does not traverse $e$ at all. Thus any $s$-$t_i$ path in 
	$H \setminus Q'$ exists in $H \setminus Q$ as well, contradicting the fact 
	that $Q$ is violating.

	Next, suppose $H \setminus Q$ has at least three components $C_1, C_2, C_3$. 
	Let $s \in C_1$. Since $H$ is connected, there exists $e \in Q$ with exactly 
	one endpoint in $C_3$. In particular, $C_1$ and $C_2$ remain disconnected
	in $H \setminus (Q \setminus \{e\})$. Since $C_2$ must contain a 
	terminal, this contradicts the assumption that $H$ satisfies 
	$(h_\ell, \ell)$-hop-connectivity. 
\end{proof}

Henceforth, we let $C_s$ denote the connected component of $H \setminus Q$ 
containing $s$, and $C_t$ the other component if it exists. 
The analysis will proceed as follows: first, we 
will bound the diameter of $C_s$ and $C_t$. Then, we will show 
that with good probability, there exists a round of Oblivious Tree Rounding 
which buys a path directly from $C_s$ to $C_t$ without using any edges in $Q$. 
 By nature of $\droute$-routers, 
we know that all tree paths have bounded hop-length, giving us the 
desired bound on the total hop length of paths from $s$ to $t_i$. 

\subsubsection{Bounding Diameter of Connected Components}

In \cite{diameter_bounds}, Chung and Garey show that for any graph $G$ with
diameter $d$, if we remove $\ell$ edges and the new graph $G'$ is connected, 
the diameter of $G'$ is at most $d(\ell + 1) + O(\ell)$. We simplify the proof 
and generalize it to bound the diameter of individual components of the new
graph (even if the new graph is no longer connected), with slightly looser 
bounds.

\begin{lemma}
\label{lem:H-component-diam}
	All connected components of $H \setminus Q$ have diameter at 
	most $O(\ell)\cdot h_\ell$.
\end{lemma}
\begin{proof}
	We prove the claim for $C_s$, as the argument for $C_t$ is analogous. 
	Let $d = \diam(C_s)$, and let $u, v \in C_s$ be
	such that $\hop(u,v) = d$. Let $V_i = \{w \in C_s: \hop(u,w) = i\}$ for 
	$i \in [d]$. Let $C_s'$ be obtained by contracting 
	each $V_i$ into a single node $v_i$ in $C_s$ and adding an isolated 
	node $t$. 
	Note that $C_s'$ is
	a path of length $d$ with an additional isolated node.
	Let $H'$ be the graph 
	obtained from $H$ by contracting each $V_i$ into a node $v_i$ 
	and contracting $C_t$ to a node $t$. 

	We will show that $H'$ has diameter at least $\frac d {4\ell+2}$. 
	Since $C_s$ was constructed from $H$ by deleting $Q$, $H'$ can be 
	constructed from $C_s'$ by adding back some subset $Q' \subseteq Q$ 
	of at most $\ell$ edges ($Q'$ may not be equal to $Q$, since 
	some edges $Q$ may have been contracted).
	Let $x_1, x_2, \dots, x_{j}$ be the subset of vertices of 
	$C_s'$ that are incident to edges in $Q'$, ordered in increasing distance 
	from $u$ in $C_s'$. Since 
	$|Q'| \leq \ell$, $j \leq 2\ell$. We partition the path $C_s'$ into 
	$j+1$ segments: 
	$u$ to $x_1$, $x_1$ to $x_2$, \dots, $x_{j-1}$ to $x_j$, $x_j$ to $v$. 
	The total hop-length of the path from $u$ to $v$ in $C_s'$ is $d$, 
	so there must be at least 
	one segment of length $\geq \frac d {j+1}$. Since this segment does not 
	have any edges in $Q'$ (aside from the endpoints), the diameter of 
	this segment in $H'$ is at least $\frac d{2(j+1)}$. Thus 
	$\diam(H') \geq \frac d {2(j+1)} \geq \frac d {4\ell + 2}$.
	Since $H'$ was constructed by contracting edges in $H$, 
	$\diam(H) \geq \diam(H')$. By Claim \ref{claim:H-diameter}, 
	$\diam(H) \leq 2h_\ell$. Thus 
	$\frac d {4\ell + 2} \leq \diam(H') \leq \diam(H) \leq 2h_\ell$, 
	so $d = \diam(C_s) \leq (8\ell + 4) h_\ell$. 
\end{proof}

\subsubsection{Bounding flow on trees}

We use a similar argument to that of Lemma \ref{lem:hc_flow_preservation} to
show that in each iteration, $y$ supports a flow of at least $\frac 1 4$ from 
$C_s$ to $C_t$ without using edges in $M^{-1}(Q)$ with constant probability. 
Crucially, we use the fact that Algorithm \ref{algo:fault_hop} scales down 
capacities $x_e$ to $O(1/\ell\sigma)$ for all $e \in Q$; thus $y$ cannot 
carry too much flow on $M^{-1}(Q)$.
Fix an iteration $j \in [\tau]$ of Algorithm \ref{algo:hop_general} 
and let $(T, M)$ be the sampled tree. 

\begin{claim}
\label{claim:flow-violated-edges}
	The expected amount of flow supported by $y$ on $M^{-1}(Q)$ is 
	$\E[y(M^{-1}(Q))] \leq \frac 1 4$.
\end{claim}
\begin{proof}
	Recall in Algorithm \ref{algo:fault_hop} that we scale down $\tilde x_e$
	values of all edges in $H$ to $\frac 1 {4\ell \sigma}$. Thus 
	$\sum_{e \in Q} \tilde x_e \leq |Q| \frac 1 {4\ell\sigma} 
	\leq \frac 1 {4\sigma}$.
	By Lemma \ref{claim:edge_load}, the expected amount of flow on $M^{-1}(Q)$
	in the tree $T$ is at most
	\[\sum_{e' \in M^{-1}(Q)} \E[y(e')] 
    \leq \sum_{e \in Q} \sum_{e': e \in M_{e'}} \E[y(e')] 
    \leq \sum_{e \in Q} \sigma \tilde x_e \leq \frac 1 4.\]
\end{proof}

\begin{lemma}
\label{lemma:good-flow-on-tree}
	With probability at least $\frac 1 3$, $y$ supports a flow of at least 
	$\frac 1 4$ from $C_s$ to $C_t$ in $T$ without using
	any edges in $M^{-1}(Q)$.
\end{lemma}
\begin{proof}
	We begin lower bounding the amount of flow supported by $y$ between 
	$C_s$ and $C_t$.
	By construction, this is at least the amount of flow 
	supported by $\tilde x$ between $C_s$ and $C_t$ on $G[V(T)]$. 
	Let $G'$ be formed from $G$ by contracting all edges in $H$; equivalently, 
	we form $G'$ by contracting $C_s$ and $C_t$ into two nodes, which we call 
	$s$ and $t$ respectively. Note that contracting $C_s$ and $C_t$ does not
	change the amount of flow supported between them. Thus, it suffices to 
	lower bound the amount of flow supported by $\tilde x$ between 
	$s$ and $t$ in $G'[V(T)]$. 

	Let $t_i$ be an arbitrary terminal in $C_t$. Let $\calP$ be the set of all
	paths $p \in \calP_i^h$ such that $p \cap Q = \emptyset$. Note that each of 
	these paths must start in $C_s$ and end in $C_t$, since $s \in C_s$ and 
	$t_i \in C_t$. Thus, since $H \setminus Q$ only has two connected 
	components, every path $p \in \calP$ must have a subpath, which we 
	denote $\psi(p)$, with one endpoint in $C_s$, the other in $C_t$, 
	and all intermediate nodes in $V \setminus H$. 
	Let $z_e = \sum_{p \in \calP, V(\phi(p)) \subseteq V(T), e \in \phi(p)} f_p^Q$. 
	Since we contracted all edges in $H$, if $e \in G'$, then $\tilde x_e = x_e$. 
	Since $x, f$ is a feasible solution to \hclp, 
	$\tilde x_e = x_e \geq z_e$ for all $e \in G'$, 
	thus we restrict our attention to $s$-$t$ flow in $G'$ supported by $z$. 
	Let $Z = \sum_{p \in \calP, V(\phi(p)) \subseteq V(T)} f_p^Q$. 
	We are left to lower bound $\E[Z]$.

	For all $p \in \calP$, $\hop(\phi(p)) \leq \hop(p) \leq h$, so 
	$|V(\phi(p))| \leq h+1$. Thus by union bound, $\Pr[V(\phi(p)) 
	\not \subseteq V(T)] \leq \sum_{v \in V(\phi(p))} \Pr[v \notin V(T)] 
	\leq \eps (h+1) \leq \frac 1 4$, where $\eps$ is the exclusion probability of 
	the $\droute$-router $\calT$. Therefore, 
	$\E[Z] = \sum_{p \in \calP} f_p^Q \Pr[V(\phi(p)) \subseteq V(T)] 
	\geq \sum_{p \in \calP} f_p^Q (1 - \frac 1 4)$. By 
	\hclp\ constraints, $\sum_{p \in \calP} f_p^Q = \sum_{p \in \calP_i^h, p \cap Q \neq \emptyset} 
	f_p^Q = 1.$ Thus $\E[Z] = \frac 3 4$. Note that the \hclp\ constraint
	also guarantees that $Z \in [0,1]$.

	Let $Z' = \max\{Z - y(M^{-1}(Q)), 0\}$. By Claim 
	\ref{claim:flow-violated-edges}, $\E[Z'] \geq \E[Z] - \E[y(M^{-1}(Q))]
	\geq \frac 3 4 - \frac 1 4 = \frac 1 2$. Since $Z \in [0,1]$, $Z' \in [0,1]$. 
	Therefore, we can apply Markov's inequality on $1 - Z'$ to get 
	$\Pr[Z' \leq \frac 1 4] = \Pr[1 - Z' \geq \frac 3 4] = 
	\Pr[1 - Z' \geq \frac 3 2 \E[1-Z']] \leq \frac 2 3$. 
	Since $Z'$ is a lower bound on the amount of flow supported by $y$ between 
	$C_s$ and $C_t$ in $T \setminus M^{-1}(Q)$, we obtain the desired 
	flow of $\frac 1 4$ with probability at least $\frac 1 3$.
\end{proof}

\subsubsection{Cost and Correctness Analysis}

\begin{claim}
\label{claim:hcfault_hopstretch}
	Fix a terminal $t_i$, a violating edge set $Q$ and an iteration $j \in [\tau]$ 
	of Algorithm \ref{algo:hop_general}. If either
	(1) $t_i$ is in the same component as $s$, or 
	(2) $F_j'$ contains a tree path from the component of $H \setminus Q$ 
		containing $t_i$ to the one containing $s$ without using any edges in $M^{-1}(Q)$,
	then $H \cup F_j \setminus Q$ contains an $s$-$t_i$ path of length at most 
	$O(\ell) h_\ell + O(\log^3n)h$.
\end{claim}
\begin{proof}
	If $t_i$ is in the same component 
	as $s$ in $H \setminus Q$, then by Lemma \ref{lem:H-component-diam}, 
	$\hop(s, t_i) \leq O(\ell)\cdot h_\ell$ in $H \setminus Q$.
	Instead, suppose that $t_i$ is not in the same component as $s$. 
	We let $C_s$ and $C_t$ denote the components of $H \setminus Q$ 
	containing $s$ and $t_i$ respectively. 
	Let $(T, M)$ be the sampled partial tree embedding, 
	and let $v_1, \dots, v_j$ be the tree path connecting $C_s$ and $C_t$ 
	bought by the algorithm. 
	By the guarantee on the $\droute$-router $\calT$, 
	$\hop_{F_j}(v_1, v_j) = \hop_G(M_{v_1v_j}) \leq O(\log^3 n/\eps) 
	= O(\log^3 n) h$. By assumption, this path does not use any edges in $Q$. 
	By Lemma \ref{lem:H-component-diam}, $\hop_{H \setminus Q}(s, v_1)$ and 
	$\hop_{H \setminus Q}(v_j, t_i)$ are both at most $O(\ell) h_\ell$, so 
	$\hop_{H \cup F_j \setminus Q}(s, t_i) \leq \hop_{H \setminus Q}(s, v_1) + 
	\hop_{F_j \setminus Q}(v_1, v_j) + \hop_{H \setminus Q}(v_j, t_i) 
	\leq O(\ell) h_\ell + O(\log^3n)h$. 
\end{proof}

\begin{lemma}
\label{lem:hcfault_correctness}
	With probability at least $\frac 1 2$, for every violating edge set $Q$,
	$H \cup F \setminus Q$ contains a path of length 
	at most $O(\ell) h_\ell + O(\log^3n)h$ from each terminal $t_i$ to $s$. 
\end{lemma}
\begin{proof}
	Fix a violating edge set $Q$. By Claim \ref{claim:H-components}, 
	$H \setminus Q$ has at most two components. Suppose
	that $H \setminus Q$ has two components $C_s$ and $C_t$. Fix an iteration 
	$j \in [\tau]$ and let $(T, M)$ be the sampled partial tree embedding.
	By Lemma \ref{lemma:good-flow-on-tree},
	$y$ supports a flow of at least $\frac 1 4$ from $C_s$ to $C_t$ in
	$T$ without using any edges in $M^{-1}(Q)$ with probability 
	at least $\frac 1 3$. By Lemma \ref{lem:setconnectivity-tree-rounding}, 
	if this occurs, the Oblivious Tree Rounding algorithm will connect
	$C_s$ and $C_t$ without using any edges in $M^{-1}(Q)$ with probability 
	at least $\phi$. Thus with probability at least $\frac \phi 3$, we buy a 
	tree path connecting $C_s$ and $C_t$ in $T \setminus M^{-1}(Q)$ 
	in iteration $j$. The probability that this does not happen in 
	any iteration is at most $(1 - \frac \phi 3)^\tau$. 
	
	By Claim \ref{claim:hcfault_hopstretch}, this implies that for a fixed $Q$, 
	the probability that $H \cup F \setminus Q$ does not 
	contain a path of hop length at most $O(\ell) h_\ell + O(\log^3n)h$ from 
	each terminal $t_i$ to $s$ is at most $(1 - \frac \phi 3)^\tau$.
	Taking a union bound over all violating edge sets $Q$, the probability 
	the algorithm fails is at most 
	$n^{2\ell} (1 - \frac \phi 3)^\tau \leq n^{2\ell} e^{-\tau\phi/3}$. 
	Setting $\tau = \frac 3 \phi \log (2n^{2\ell})$ gives us the desired 
	failure probability of $\frac 1 2$.
\end{proof}

\begin{lemma}
\label{lem:hcfault_cost}
	The total expected cost for Algorithm \ref{algo:fault_hop} is at most 
	$O(\ell \log^6 n) \sum_{e \in E} c(e)x_e$.
\end{lemma}
\begin{proof}
	By Lemma \ref{lem:setconnectivity-tree-rounding}, for any $e' \in E(T)$
	and any $j \in [\tau]$
	the probability $e' \in F_j'$ is at most 
	$O(4 \cdot \textnormal{height}(T) \log^2n) y(e')$. Since we can assume 
	the height of $T$ is $O(\log n)$ (see discussion in Section~\ref{sec:prelim}), 
	$\Pr[e' \in F_j'] \leq O(\log^3n) y(e')$.
	By Lemma \ref{lem:hc_cost_general} with 
	$\gamma = O(\log^3 n), \tau = O(\ell \log n)$, 
	$\E[c(F)] \leq O(\ell \log^6 n) \sum_{e \in E} c(e) \tilde x_e$. 
	Since $\tilde x_e < x_e$ for all $e \in E \setminus H$, 
	$\E[c(F)] \leq O(\ell \log^6 n) \sum_{e \in E} c(e) x_e$.

	Since all edges in \LG~have large $x_e$ value, 
	$c(\LG) = \sum_{e : x_e > \frac 1 {4\ell\sigma}} c(e)
	\leq 4\ell \sigma \sum_{e \in E} c(e)x_e
	\leq O(\ell (\log n \log \frac {\log n}\eps)) \sum_{e \in E} c(e)x_e
	\leq O(\ell \log^2 n) \sum_{e \in E} c(e)x_e$ since $\eps = O(1/h)$ 
	and $h \leq n$. Thus the total expected cost is at most 
	$O(\ell \log^6n + \ell \log^2n) \sum_{e \in E} c(e)x_e 
	= O(\ell \log^6n) \sum_{e \in E} c(e)x_e$.   
\end{proof}

\begin{proof}[Proof of Theorem \ref{thm:hcfault_main}]
	We run Algorithm \ref{algo:fault_hop} for $\ell = 1, \dots, k$
	until we obtain the desired $k$-connectivity. First, note that $h_0 = h$, 
	and by Lemma \ref{lem:hcfault_correctness}, 
	$h_{\ell+1} \leq O(\ell) h_{\ell} + O(\log^3n) h$. By a simple inductive 
	argument, $h_k \leq O(k^k \log^3 n) \cdot h$. 
	Thus the hop approx is $\beta = O(k^k \log^3 n)$. 
	By Lemma \ref{lem:hcfault_cost}, the expected cost when $k_i = k$ 
	for all terminals is 
	$\sum_{\ell = 1}^k O(\ell \log^6n) \sum_{e \in E} c(e)x_e 
	\leq O(k^2 \log^6n) \sum_{e \in E} c(e)x_e$. 
	In general, let $T_j = \{t_i: k_i = j\}$. 
	We solve the single source problem separately
	for each $T_j$, $j \in [k]$ for a total expected cost approx 
	$\alpha = O(k^3 \log^6n)$.
\end{proof}

\begin{remark}
	The main difficulty in extending this argument to the multicommodity 
	setting lies in the fact that the components of $H$ may have 
	large diameter. Therefore, this approach cannot be directly used to show that 
	the new paths between terminal pairs have bounded hop length.
\end{remark}

\subsection{Multicommodity Fault-HCND when $k=2$}
\label{sec:junction}

We consider Fault-HCND in the multicommodity setting iwth $k = 2$. 
Recall that given terminal pairs $s_i, t_i$ for $i \in [r]$, 
our goal is to obtain a subgraph $F \subseteq E$ such that for all $e \in F$, 
$i \in [r]$, $F \setminus \{e\}$ contains an $s_i$-$t_i$ path with hop length 
at most $h$. We solve this via reduction to single-source and prove the following
lemma:
\begin{lemma}
\label{lem:junction_reduction}
  Suppose we are given an $(\alpha, \beta)$-approximation for single-source
  $(h,2)$-Fault-HCND with respect to $\opt_{LP}$. Then, there exists 
  an $(O(\alpha \log^2 r), O(\beta\log r))$-approximation for 
  $(h,2)$-Fault-HCND in the general multicommodity setting with respect 
  to $\opt_{I}$.
\end{lemma}

Let $\calT = \cup_{i \in [r]}\{s_i, t_i\}$ denote the set of all terminals (so 
$|\calT| = 2r$). We employ the concept of a \emph{junction structure} with 
root $u$ and hop length $h$, which we define as 
a subgraph $H$ in which for all terminals $t \in H \cap \calT$, 
$u$ and $t$ are 
$(h,2)$-hop-connected. The \emph{density} of $H$ is
the ratio of its cost to the number of terminal pairs covered, i.e.
$c(H)/|\{i: s_i, t_i \in H\}|$. 

The goal of this section is to prove the following two 
lemmas: the first showing that a good junction structure exists, and the 
second providing an efficient algorithm to find it. Recall that 
$\opt_{I}$ denotes the optimal \emph{integral} solution to the given 
Fault-HCND instance.

\begin{lemma}
\label{lem:junction_main_exist}
	There exists a junction structure with hop length $4h(\log r + 1)$ 
	and density $O(1/r) \cdot \opt_{I}$.
\end{lemma}

\begin{lemma}
\label{lem:junction_main_algo}
	Suppose we are given an $(\alpha, \beta)$-approximation for single-source 
	$(h',2)$-Fault-HCND with respect to the LP. Then, there exists an 
	algorithm that gives a junction structure with hop length
	$\beta h'$ and density at most
	$O(\alpha \log r)$ times the optimal density of a junction structure with 
	hop length $h'$. 
\end{lemma}

We will prove Lemma \ref{lem:junction_main_exist} in Section \ref{sec:junction_exist}
and Lemma \ref{lem:junction_main_algo} in Section \ref{sec:junction:algo}.
Assuming these lemmas, we complete the proof of Lemma \ref{lem:junction_reduction}:
\begin{proof}[Proof of Theorem \ref{lem:junction_reduction}]
	Combining Lemmas \ref{lem:junction_main_exist} and \ref{lem:junction_main_algo},
	there exists an algorithm to obtain a junction structure of density at most 
	$O\left(\frac {\alpha \log r} r\right) \opt_{I}$.
	We employ an iterative greedy approach. 
	Let $I = [r]$. At each step $j$, we apply the algorithm on
	graph $G$ with terminal pairs $I_j$ to obtain a junction structure $H_j$. 
	We set $I_{j+1} = I_j \setminus \{i \in I_j: s_i,t_i \in H_j\}$ and repeat 
	until there are no remaining terminals. We return 
	$H = \cup_{j} H_j$. This terminates after at most $r$ steps, 
	since we cover at least one terminal pair in each iteration. 
	The total cost is at most 
	$\sum_{j \in [r]} c(H_j) 
	= \sum_{j \in [r]} \density(H_j) (|I_j| - |I_{j+1}|)
	= O(\alpha \log r) \opt \cdot \sum_{j \in [r]} \frac{|I_j| - |I_{j+1}|}{|I_j|}
	\leq O(\alpha \log^2r)\opt.$

	To prove the hop-bound, fix $i \in [r]$. Let $H_j$ be 
	the junction structure containing $s_i, t_i$, and let $u$ be its root. 
	Then for any 
	$e \in H$, $H_j \setminus \{e\}$ contains an 
	$s_i$-$u$ path $p$ and a $t_i$-$u$ path $p'$, each of hop length at most 
	$\beta h'$. Concatenating $p$ and $p'$ gives us an $s_i$-$t_i$ path of 
	hop length at most $2\beta h'$ in $H_j \setminus \{e\}$, which is a subset of 
	$H \setminus \{e\}$. Thus every terminal pair $s_i, t_i$ is 
	$(2\beta h', 2)$-hop-connected in $H$. Since $2\beta h = O(\log r \beta) h$,
	we get our desired bound.
\end{proof}

\remark{It is not difficult to see that the single-source reduction to 
obtain a good junction structure if one exists can be extended to $k > 2$. 
The bottleneck in using this approach for higher connectivity is 
showing the existence of a good junction structure.}

\subsubsection{Existence of good junction structure}
\label{sec:junction_exist}

Let $E'$ denote an optimal solution 
to the instance of $(h, 2)$-Fault-HCND, so $\opt_{I} = c(E')$. We will 
partition $E'$ into several junction structures, that, in 
total, cover a large fraction of the terminal pairs. For each 
$i \in [r]$, let $C_i \subseteq E'$ denote a minimal set of edges 
such that $s_i$ and $t_i$ are $(h,2)$-hop-connected. 
We use the following properties of $C_i$.

\begin{claim}
\label{claim:junction_Ci_paths}
	For all $i \in [r]$, all $v \in C_i$, there exists an $s_i$-$t_i$ path 
	of hop length at most $h$ that contains $v$.
\end{claim}
\begin{proof}
	The claim follows by minimality of $C_i$, since if $v \notin p$ for all 
	$p \in \calP_i^h$, then $s_i$ and $t_i$ would remain $(h,2)$-hop-connected 
	in $C_i \setminus \{v\}$.
\end{proof}

\begin{claim}
\label{claim:junction_Ci_diameter}
	For all $i \in [r]$, the diameter of $C_i$ is at most $h$. 
\end{claim}
\begin{proof}
	Fix $i \in [r]$. Let $v, w$ be two arbitrary vertices in $C_i$. 
	Let $a = \hop(s_i, v)$, $b = \hop(s_i, w)$. As a simple corollary of 
	Claim \ref{claim:junction_Ci_paths}, $\hop(s_i, v) + \hop(t_i, v) \leq h$.
	Thus $\hop(t_i, v) \leq h - a$ and $\hop(t_i, w) \leq h-b$. 
	Since we can connect $v$ to $w$ by going through $s_i$ or $t_i$, 
	$\hop(v, w) \leq \min(a+b, (h-a) + (h-b)) = \min(a+b, 2h-(a+b)) \leq h$.
	Therefore, every pair of nodes has a path in $C_i$ of hop length at most $h$.
\end{proof}

\begin{claim}
\label{claim:junction_Ci_2conn}
  For all $i \in [r]$, $C_i$ is 2-edge-connected.
\end{claim}
\begin{proof}
	Fix $i \in [r]$. 
	We will show that for every $v \in C_i$, $v$ is 2-edge-connected to $s_i$. 
	The claim then follows from transitivity of edge-connectivity. 
	Let $v \in C_i$, and let $p$ be the $s_i$-$t_i$ path containing $v$ 
	given by Claim \ref{claim:junction_Ci_paths}. Let $e \in C_i$. 
	If $e \notin p$, then $v$ is connected to $s_i$ in $C_i \setminus \{e\}$. 
	Similarly, if $e \in p$ but $e$ is further on $p$ from $s_i$ than $v$, 
	then $v$ is connected to $s_i$ in $C_i \setminus \{e\}$. 
	Lastly, suppose $e \in p$ and $e$ is closer to $s_i$ than $v$ is. Then, 
	$v$ must be connected to $t_i$ in $C_i \setminus \{e\}$. By 
	$(h,2)$-hop-connectivity of $s_i$ and $t_i$, there exists some path 
	$p'$ connecting $s_i$ to $t_i$ in $C_i \setminus \{e\}$. Therefore, 
	$v$ is still connected to $s_i$ in $C_i \setminus \{e\}$ by concatenating 
	the paths $p$ and $p'$. Thus $v$ is 2-edge-connected to $s_i$ as desired.
\end{proof}

We consider balls centered at terminals.
Let $B(u, \ell) = \{v \in G: \hop_G(u,v) \leq \ell\}$, and let 
$B'(u, \ell)$ denote the maximal 2-edge-connected component of 
$G[B(u, \ell)]$ containing $u$. 

\begin{claim}
\label{claim:junction_ball_radius}
	For every $v \in B'(u, \ell)$, there exists a $v$-$u$ path of hop length 
	at most $\ell$ in $B'(u, \ell)$. 
\end{claim}
\begin{proof}
	Let $v \in B'(u, \ell)$. Since $B'(u, \ell) \subseteq B(u, \ell)$, 
	there exists some $u$-$v$ path $p$ of hop length at most $\ell$ in $G$.
	It is clear by definition that $p \subseteq B(u, \ell)$.  
	We will show that $p \subseteq B'(u, \ell)$. 
	Let $w \in p$. We will show that $w$ is 2-connected to $u$ in $B'(u, \ell)$.
	To do so, let $e \in B'(u, \ell)$. Let $p'$ denote the subset of $p$ 
	from $u$ to $w$. If $e \notin p'$, then $w$ and $u$ are connected in 
	$B'(u, \ell) \setminus \{e\}$ via $p'$. If $e \in p'$, then $w$ must 
	still be connected to $v$ via the section of $p$ from $w$ to $v$. 
	Since $v \in B'(u, \ell)$, there exists some path $p''$ in 
	$B'(u, \ell) \setminus \{e\}$ connecting $u$ to $v$. Thus $w$ is connected 
	to $u$ via $v$ in $B'(u, \ell) \setminus \{e\}$. Since $w$ was arbitrary, 
	all nodes in $p$ must be 2-connected to $u$, so $p \subseteq B'(u, \ell)$. 
\end{proof}

\remark{Claims \ref{claim:junction_Ci_2conn} and \ref{claim:junction_ball_radius}
rely on properties of 2-connectivity that do not generalize to higher $k$, 
limiting us from using this type of junction-based argument when $k > 2$}.

We say $B'(u, \ell)$ \emph{captures} $i \in [r]$ 
if $C_i \subseteq B'(u, \ell)$. We say $B'(u, \ell)$ \emph{intersects} 
$i \in [r]$ if $i$ is not captured but there is some node 
$v \in C_i \cap B'(u, \ell)$.
We start by considering balls of the form $B(s_1, hp)$, where $p \in \N$.
For ease of notation, we let $s$ denote $s_1$.

\begin{lemma}
\label{lem:junction_captured_terminals}
	Suppose $i \in [r]$ is intersected by $B'(s, hp)$ for some $p \in \N$. 
	Then $i$ is captured by $B'(s, h(p+1))$. 
\end{lemma}
\begin{proof}
	Since $i$ is intersected by $B'(s, hp)$, there exists some 
	$v \in C_i \cap B'(s, hp)$. Let $w \in C_i$. 
	By Claim \ref{claim:junction_Ci_diameter}, $\hop(w,v) \leq h$. 
	Since $v \in B'(s, hp)$, by Claim \ref{claim:junction_ball_radius}, 
	$\hop(v, s) \leq hp$. 
	Thus $\hop(w, s) \leq \hop(w,v) + \hop(v, s) \leq h(p+1)$. Since 
	$w \in C_i$ was arbitrary, $C_i \subseteq B(s, h(p+1))$.
	To show that $C_i \subseteq B'(s, h(p+1))$, consider any 
	$w \in C_i$. By Claim \ref{claim:junction_Ci_2conn}, $w$ is 2-connected 
	to $v$ in $G[C_i] \subseteq G[B(s, h(p+1))]$. Since $v \in 
	B'(s, hp)$, $v$ is 2-connected to $s$ in $G[B(s, hp)] 
	\subseteq G[B(s, h(p+1))]$. Thus by transitivity of edge-connectivity, 
	$w$ and $s$ are 2-connected in $G[B(s, h(p+1))]$, so 
	$w \in B'(s, h(p+1))$. Since $w$ was arbitrary, $i$ 
	is captured by $B'(s, h(p+1))$.
\end{proof}

\begin{claim}
\label{claim:junction_good_ball}
	There exists some $p \leq \log_2 r+ 1$ where $B'(s, hp)$ captures 
	at least as many pairs as it intersects.  
\end{claim}
\begin{proof}
	By Lemma \ref{lem:junction_captured_terminals}, if $B'(s, hp)$ 
	captures less pairs than it intersects, then the number of 
	captured pairs in $B'(s, h(p+1))$ is at least double the number of 
	captured pairs in $B'(s, hp)$. Furthermore, $B'(s, h)$ captures 
	at least one pair (namely, $i = 1$), since Claims \ref{claim:junction_Ci_2conn}
	and \ref{claim:junction_Ci_diameter} tell us that 
	$C_1 \subseteq B'(s_1, h)$. 
	Therefore, if there is no $p \leq \log_2 r$ 
	in which $B'(s, hp)$ captures at least as many pairs as it intersects, 
	then $B'(s, h(\log_2 r + 1))$ must have captured all the pairs.
\end{proof}

\begin{lemma}
\label{lem:junction_validity}
	$B'(s, hp)$ is a valid junction structure with root $s$ and hop 
	length $4hp$. 
\end{lemma}
\begin{proof}
	By Claim \ref{claim:junction_ball_radius}, the diameter of $B'(s, hp)$
	is at most $2hp$. Consider any edge $e \in B'(s, hp)$. Since removing 
	an edge from a 2-connected graph can increase its diameter by at most 2
	(see \cite{diameter_bounds}), $B'(s, hp) \setminus \{e\}$ has diameter 
	at most $4hp$. Thus all terminals have a path to $s$ in 
	$B'(s, hp) \setminus \{e\}$ of hop-length at most $4hp$. 
\end{proof}

\begin{proof}[Proof of Lemma \ref{lem:junction_main_exist}]
By Claim \ref{claim:junction_good_ball}, we can let $H_1^* = B'(s, hp)$ for 
some $p$ such that $B'(s, hp)$ captures at least as many pairs as it intersects.
Then, we remove all pairs that $H_1^*$ captures or intersects, remove all 
edges in $B'(s, hp)$, and repeat
the process on the remaining terminals (choosing $s$ to be some other remaining
terminal) and remaining graph. Notice that terminals that are not intersected 
by $B'(s, hp)$ are not affected by its removal, since all of $C_i$ is still 
contained in the remaining graph. We continue this process until all terminal 
pairs have been either captured or intersected. This gives us junction structures
$H_1^*, H_2^*, \dots$ with $\sum_{j} c(H_j^*) \leq c(E') = \opt_{I}$. 
Since we capture at least as many terminal pairs as we intersect in 
each iteration, the total number of captured 
terminals is at least $\frac r 2$. Since the average density of all 
junction structures is at most $2\opt_{I}/r$, there must be some 
$H_{\tau}^*$ that has density at most $2\opt_{I}/r$. 
By Lemma \ref{lem:junction_validity} and Claim \ref{claim:junction_good_ball}
$H_{\tau}^*$ is a valid junction structure with  hop length at most 
$4h(\log r + 1)$. 
\end{proof}

\subsubsection{Finding a good junction structure}
\label{sec:junction:algo}

We start by guessing the root $u$ of the junction; we can run this algorithm 
for every possible root $u \in V$ and take the resulting junction structure with 
minimum density. Then, we solve the following LP relaxation for the min-density 
junction structure rooted at $u$. This relaxation has the same structure as that 
of the \hclp~ in the fault-tolerant, single source setting, with the additional 
variable $y_t$ for each terminal $t \in \calT$ 
denoting whether or not $t$ is included in the junction structure. We normalize 
$\sum_{t} y_t = 1$. For each terminal $t$, we let $P_t^{h'}$ denote the set of 
all $t$-$u$ paths in $G$ of hop-length at most $h'$. We write the LP for 
general $k$, and note that when $k = 2$, all violated edge sets $Q$ consist of 
a single edge.
We refer to this LP relaxation as \junctionlp.

\begin{subequations}
\begin{align}
  \min \sum_{e \in E} c(e)x_e &\\
  \text{s.t.} \sum_{p \in \calP^{h'}_t, p \cap Q = \emptyset} f_p^Q &= y_t 
  &\forall t \in \calT, \forall Q \subseteq E, |Q| < k \\
  \sum_{p \in \calP^{h'}_t, e \in p} f_p^Q &\leq x_e 
  &\forall e \in E, \forall t \in \calT, \forall Q \subseteq E, |Q| < k \\
  y_{s_i} &= y_{t_i} &\forall i \in [r] \\
  \sum_{t \in \calT} y_t &= 1 \\
  x_e, f_p, y_t &\geq 0
\end{align}
\end{subequations}

Let $x^*, f^*, y^*$ be an optimal solution to this LP. Let $y_{\max}^* = 
\max_{t \in \calT} y_t^*$. Partition the set of 
terminals $\calT$ into $\log r+1$ groups $T_1, \dots, T_{\log r + 1}$, where 
$t \in T_j$ if 
$\frac {y_{\max}^*} {2^{j+1}} < y_t^* \leq \frac {y_{\max}^*} {2^j}$.

\begin{claim}
\label{claim:junction_goodgroup}
	There exists $\theta$ such that 
	$\sum_{t \in T_{\theta}} y_t^* \geq 1/(2(\log r+1))$. 
\end{claim}
\begin{proof}
	Suppose $t \notin T_j$ for $j \in [\log r]$. Then, $y_t^* 
	\leq y_{\max}^*/(2^{\log r +2}) = y_{\max}^*/4r$. 
	Thus $\sum_{t \notin \cup T_j} y_t^* \leq |\calT|y_{\max}^*/4r 
	\leq y_{\max}^*/2 \leq \frac 1 2$, since all $y_t^* \leq 1$ by 
	constraint (e) of \junctionlp. This implies that $\sum_{t \in \cup T_j} y_t^* 
	= \sum_{t \in \calT} y_t^* - \sum_{t \notin \cup T_j} y_t^* 
	\geq \frac 1 2$. Since there are $\log r +1$ groups, at least
	one group must have total $y_t^*$ sum at least $1/(2(\log r + 1))$.
\end{proof}

We run the $(\alpha, \beta)$-approximation for single-source 
$(h', 2)$-Fault-HCND with root $u$ and terminals $T_{\theta}$, where 
$\theta$ is given by Claim \ref{claim:junction_goodgroup}, and return 
the resulting solution $H$.
The following lemma follows immediately from the definition of junction 
structures:

\begin{lemma}
\label{lem:junction_hopstretch}
	$H$ is a junction structure with root $u$, terminals $T_\theta$, 
	and hop length $\beta h'$.
\end{lemma}

\begin{claim}
\label{claim:junction_ss_feasible}
	Let $\lambda = \frac{2^{\theta + 1}}{y_{\max}^*}$, $x = \lambda x^*$, 
	$f = \lambda f^*$. Then $(x, f)$ is a feasible solution to 
	\hclp~for single source $(h', 2)$-Fault-HCND 
	with terminals $T_\theta$.
\end{claim}
\begin{proof}
The fact that $(x, f)$ satisfies the second constraint of \hclp~ follows immediately 
from the fact that $(x^*, f^*)$ satisfies constraint (c) of \junctionlp,
since both $x^*$ and $f^*$ are scaled by the same value. 

We then show that $(x, f)$ satisfies constraint the first constraint of \hclp, that is,
for all $t \in T_\theta$, for all $q \in E$,
$\sum_{p \in \calP_{t}^{h'}, p \cap Q = \emptyset} f_p^{Q} = 1$. 
By constraint (b) of \junctionlp, we know that
$\sum_{p \in \calP_{t}^{h'}, p \cap Q = \emptyset} f_p^{Q*} = y_t^*$, so 
$\sum_{p \in \calP_{t}^{h'}, p \cap Q = \emptyset} f_p^Q
= \lambda \sum_{p \in \calP_{t}^{h'}, p \cap Q = \emptyset} f_p^{Q*} 
= \lambda y_t^*$. Since $t \in T_\theta$, 
$y_t^* > \frac{y_{\max}^*}{2^{\theta + 1}}$, so $\lambda y_t^* > 1$. 
We can scale down $f_p$ values to get the sum to be exactly $1$ without violating 
any constraints.
\end{proof}

\begin{lemma}
\label{lem:junction_density_lp}
  The density of $H$ is at most $O(\alpha \log r) \sum_{e \in E} c(e)x_e^*$.
\end{lemma}
\begin{proof}
	Since $H$ is given by an $(\alpha, \beta)$-approximation for single-source 
	$(h', 2)$-Fault-HCND, $c(H) \leq \alpha \cdot \opt_{HC}$, where $\opt_{HC}$ 
	is the optimal LP solution for the given single-source instance. By 
	Claim \ref{claim:junction_ss_feasible}, there exists a feasible solution 
	to $\hclp$ with cost $\frac{2^{\theta + 1}}{y_{\max}^*} x^*$, so 
	$\opt_{HC} \leq \frac{2^{\theta + 1}}{y_{\max}^*} \sum_{e \in E} c(e)x_e^*$. 
	To lower bound the number of terminals covered by $H$, notice that since 
	$y^*$ is a feasible solution to \junctionlp, $y_{s_i} = y_{t_i}$ for all 
	terminal pairs $i \in [r]$. Therefore, for every terminal in $T_\theta$, 
	its corresponding terminal pair is also in $T_\theta$, so 
	$|\{i: s_i, t_i \in H\}| \geq |T_\theta|/2$. 
	Thus the density of $H$ is at most 
	$(4\alpha 2^{\theta})/(y_{\max}^*|T_{\theta}|) \sum_{e \in E} c(e)x_e^*$.
	Recall that by our choice of $\theta$ (see Claim \ref{claim:junction_goodgroup}),
	$|T_\theta| y_{\max}^*/2^{\theta} \geq \sum_{t \in T_\theta} y_t^*
	\geq \frac 1 {2(\log r + 1)}$. Therefore, 
	$(4\alpha 2^{\theta})/(y_{\max}^*|T_{\theta}|) \leq O(\alpha \log r)$ 
	as desired. 
\end{proof}

\begin{proof}[Proof of Lemma \ref{lem:junction_main_algo}]
	Let $H^*$ be a junction structure rooted at $u$ of optimal density and let 
	$T^*$ be the set of terminal pairs it captures. 
	Set $x_e = 1$ for $e \in H^*$, $y_t = 1$ for all $t \in T^*$. 
	For every captured terminal $t$ and every $Q \subseteq E$, $|Q| < k$,
	choose some $t$-$u$ path $p \in H^* \setminus Q$ of hop length at most $h$ 
	and set $f_p^{Q} = 1$.

	It is easy to verify that $\frac 1 {|T^*|}(x, f, y)$ satisfies 
	all constraints of \junctionlp~. Thus $\sum_{e \in E} c(e) x_e^*
	\leq c(H^*)/|T^*| = \density(H^*)$. Combining this with 
	Lemma \ref{lem:junction_density_lp},
	$\density(H) \leq O(\alpha \log r) \density(H^*)$.
	This, along with Lemma \ref{lem:junction_hopstretch},
	concludes the proof.
\end{proof}

\mypara{Acknowledgements:} We thank Mik Zlatin for helpful discussions 
on hop-constrained metric embeddings. 

\bibliographystyle{plainurl}
\bibliography{references}

\begin{thebibliography}{10}

\bibitem{AhmadiGHJM-prize}
Ali Ahmadi, Iman Gholami, MohammadTaghi Hajiaghayi, Peyman Jabbarzade, and
  Mohammad Mahdavi.
\newblock {\em 2-Approximation for Prize-Collecting Steiner Forest}, pages
  669--693.
\newblock URL: \url{https://epubs.siam.org/doi/abs/10.1137/1.9781611977912.25},
  \href
  {http://arxiv.org/abs/https://epubs.siam.org/doi/pdf/10.1137/1.9781611977912.25}
  {\path{arXiv:https://epubs.siam.org/doi/pdf/10.1137/1.9781611977912.25}},
  \href {https://doi.org/10.1137/1.9781611977912.25}
  {\path{doi:10.1137/1.9781611977912.25}}.

\bibitem{aaabn06}
Noga Alon, Baruch Awerbuch, Yossi Azar, Niv Buchbinder, and Joseph Naor.
\newblock A general approach to online network optimization problems.
\newblock {\em ACM Transactions on Algorithms (TALG)}, 2(4):640--660, 2006.

\bibitem{althaus05}
Ernst Althaus, Stefan Funke, Sariel Har-Peled, Jochen Könemann, Edgar~A.
  Ramos, and Martin Skutella.
\newblock Approximating k-hop minimum-spanning trees.
\newblock {\em Operations Research Letters}, 33(2):115--120, 2005.
\newblock URL:
  \url{https://www.sciencedirect.com/science/article/pii/S0167637704000719},
  \href {https://doi.org/https://doi.org/10.1016/j.orl.2004.05.005}
  {\path{doi:https://doi.org/10.1016/j.orl.2004.05.005}}.

\bibitem{andrews04}
M.~Andrews.
\newblock Hardness of buy-at-bulk network design.
\newblock In {\em 45th Annual IEEE Symposium on Foundations of Computer
  Science}, pages 115--124, 2004.
\newblock \href {https://doi.org/10.1109/FOCS.2004.32}
  {\path{doi:10.1109/FOCS.2004.32}}.

\bibitem{az98}
Matthew Andrews and Lisa Zhang.
\newblock The access network design problem.
\newblock In {\em Proceedings of the 39th Annual Symposium on Foundations of
  Computer Science}, FOCS '98, page~40, USA, 1998. IEEE Computer Society.

\bibitem{acsz11}
Spyridon Antonakopoulos, Chandra Chekuri, Bruce Shepherd, and Lisa Zhang.
\newblock Buy-at-bulk network design with protection.
\newblock {\em Mathematics of Operations Research}, 36(1):71--87, 2011.
\newblock \href {http://arxiv.org/abs/https://doi.org/10.1287/moor.1110.0484}
  {\path{arXiv:https://doi.org/10.1287/moor.1110.0484}}, \href
  {https://doi.org/10.1287/moor.1110.0484} {\path{doi:10.1287/moor.1110.0484}}.

\bibitem{ArslanJL20}
Okan Arslan, Ola Jabali, and Gilbert Laporte.
\newblock A flexible, natural formulation for the network design problem with
  vulnerability constraints.
\newblock {\em INFORMS Journal on Computing}, 32(1):120--134, 2020.
\newblock \href {http://arxiv.org/abs/https://doi.org/10.1287/ijoc.2018.0869}
  {\path{arXiv:https://doi.org/10.1287/ijoc.2018.0869}}, \href
  {https://doi.org/10.1287/ijoc.2018.0869} {\path{doi:10.1287/ijoc.2018.0869}}.

\bibitem{acm81}
G.~R. Ash, R.~H. Cardwell, and R.~P. Murray.
\newblock Design and optimization of networks with dynamic routing.
\newblock {\em The Bell System Technical Journal}, 60(8):1787--1820, 1981.
\newblock \href {https://doi.org/10.1002/j.1538-7305.1981.tb00297.x}
  {\path{doi:10.1002/j.1538-7305.1981.tb00297.x}}.

\bibitem{awerbuch97}
B.~Awerbuch and Y.~Azar.
\newblock Buy-at-bulk network design.
\newblock In {\em Proceedings 38th Annual Symposium on Foundations of Computer
  Science}, pages 542--547, 1997.
\newblock \href {https://doi.org/10.1109/SFCS.1997.646143}
  {\path{doi:10.1109/SFCS.1997.646143}}.

\bibitem{ba92}
Anantaram Balakrishnan and Kemal Altinkemer.
\newblock Using a hop-constrained model to generate alternative communication
  network design.
\newblock {\em ORSA Journal on Computing}, 4(2):192--205, 1992.
\newblock \href {http://arxiv.org/abs/https://doi.org/10.1287/ijoc.4.2.192}
  {\path{arXiv:https://doi.org/10.1287/ijoc.4.2.192}}, \href
  {https://doi.org/10.1287/ijoc.4.2.192} {\path{doi:10.1287/ijoc.4.2.192}}.

\bibitem{BalakrishnanK17}
Anantaram Balakrishnan and Christian~Vad Karsten.
\newblock {Container shipping service selection and cargo routing with
  transshipment limits}.
\newblock {\em European Journal of Operational Research}, 263(2):652--663,
  2017.
\newblock URL:
  \url{https://ideas.repec.org/a/eee/ejores/v263y2017i2p652-663.html}, \href
  {https://doi.org/10.1016/j.ejor.2017.05.03}
  {\path{doi:10.1016/j.ejor.2017.05.03}}.

\bibitem{bkp01}
Judit Bar-Ilan, Guy Kortsarz, and David Peleg.
\newblock Generalized submodular cover problems and applications.
\newblock {\em Theoretical Computer Science}, 250(1):179--200, 2001.
\newblock URL:
  \url{https://www.sciencedirect.com/science/article/pii/S0304397599001309},
  \href {https://doi.org/https://doi.org/10.1016/S0304-3975(99)00130-9}
  {\path{doi:https://doi.org/10.1016/S0304-3975(99)00130-9}}.

\bibitem{bartal_tree}
Yair Bartal.
\newblock On approximating arbitrary metrices by tree metrics.
\newblock In {\em Proceedings of the Thirtieth Annual ACM Symposium on Theory
  of Computing}, STOC '98, page 161–168, New York, NY, USA, 1998. Association
  for Computing Machinery.
\newblock \href {https://doi.org/10.1145/276698.276725}
  {\path{doi:10.1145/276698.276725}}.

\bibitem{bgsw93}
Daniel Bienstock, Michel~X Goemans, David Simchi-Levi, and David Williamson.
\newblock A note on the prize collecting traveling salesman problem.
\newblock {\em Mathematical programming}, 59(1-3):413--420, 1993.

\bibitem{BienstockG96}
Daniel Bienstock and Oktay G\"{u}nl\"{u}k.
\newblock Capacitated network design—polyhedral structure and computation.
\newblock {\em INFORMS Journal on Computing}, 8(3):243--259, 1996.
\newblock \href {http://arxiv.org/abs/https://doi.org/10.1287/ijoc.8.3.243}
  {\path{arXiv:https://doi.org/10.1287/ijoc.8.3.243}}, \href
  {https://doi.org/10.1287/ijoc.8.3.243} {\path{doi:10.1287/ijoc.8.3.243}}.

\bibitem{Botton13}
Quentin Botton, Bernard Fortz, Luis Gouveia, and Michael Poss.
\newblock Benders decomposition for the hop-constrained survivable network
  design problem.
\newblock {\em INFORMS Journal on Computing}, 25(1):13--26, 2013.
\newblock \href {http://arxiv.org/abs/https://doi.org/10.1287/ijoc.1110.0472}
  {\path{arXiv:https://doi.org/10.1287/ijoc.1110.0472}}, \href
  {https://doi.org/10.1287/ijoc.1110.0472} {\path{doi:10.1287/ijoc.1110.0472}}.

\bibitem{ByrkaGRS13}
Jaros{\l}aw Byrka, Fabrizio Grandoni, Thomas Rothvo{\ss}, and Laura Sanit{\`a}.
\newblock Steiner tree approximation via iterative randomized rounding.
\newblock {\em Journal of the ACM (JACM)}, 60(1):1--33, 2013.

\bibitem{CGL15}
Parinya Chalermsook, Fabrizio Grandoni, and Bundit Laekhanukit.
\newblock {\em On Survivable Set Connectivity}, pages 25--36.
\newblock URL: \url{https://epubs.siam.org/doi/abs/10.1137/1.9781611973730.3},
  \href
  {http://arxiv.org/abs/https://epubs.siam.org/doi/pdf/10.1137/1.9781611973730.3}
  {\path{arXiv:https://epubs.siam.org/doi/pdf/10.1137/1.9781611973730.3}},
  \href {https://doi.org/10.1137/1.9781611973730.3}
  {\path{doi:10.1137/1.9781611973730.3}}.

\bibitem{bab_random_greedy}
Moses Charikar and Adriana Karagiozova.
\newblock On non-uniform multicommodity buy-at-bulk network design.
\newblock In {\em Proceedings of the Thirty-Seventh Annual ACM Symposium on
  Theory of Computing}, STOC '05, page 176–182, New York, NY, USA, 2005.
  Association for Computing Machinery.
\newblock \href {https://doi.org/10.1145/1060590.1060617}
  {\path{doi:10.1145/1060590.1060617}}.

\bibitem{chks09}
C.~Chekuri, M.~T. Hajiaghayi, G.~Kortsarz, and M.~R. Salavatipour.
\newblock Approximation algorithms for nonuniform buy-at-bulk network design.
\newblock {\em SIAM Journal on Computing}, 39(5):1772--1798, 2010.
\newblock \href {http://arxiv.org/abs/https://doi.org/10.1137/090750317}
  {\path{arXiv:https://doi.org/10.1137/090750317}}, \href
  {https://doi.org/10.1137/090750317} {\path{doi:10.1137/090750317}}.

\bibitem{cegs11}
Chandra Chekuri, Guy Even, Anupam Gupta, and Danny Segev.
\newblock Set connectivity problems in undirected graphs and the directed
  steiner network problem.
\newblock {\em ACM Trans. Algorithms}, 7(2), mar 2011.
\newblock \href {https://doi.org/10.1145/1921659.1921664}
  {\path{doi:10.1145/1921659.1921664}}.

\bibitem{cj23}
Chandra Chekuri and Rhea Jain.
\newblock {Approximation Algorithms for Network Design in Non-Uniform Fault
  Models}.
\newblock In Kousha Etessami, Uriel Feige, and Gabriele Puppis, editors, {\em
  50th International Colloquium on Automata, Languages, and Programming (ICALP
  2023)}, volume 261 of {\em Leibniz International Proceedings in Informatics
  (LIPIcs)}, pages 36:1--36:20, Dagstuhl, Germany, 2023. Schloss Dagstuhl --
  Leibniz-Zentrum f{\"u}r Informatik.
\newblock URL: \url{https://drops.dagstuhl.de/opus/volltexte/2023/18088}, \href
  {https://doi.org/10.4230/LIPIcs.ICALP.2023.36}
  {\path{doi:10.4230/LIPIcs.ICALP.2023.36}}.

\bibitem{ckn01}
Chandra Chekuri, Sanjeev Khanna, and Joseph~(Seffi) Naor.
\newblock A deterministic algorithm for the cost-distance problem.
\newblock In {\em Proceedings of the Twelfth Annual ACM-SIAM Symposium on
  Discrete Algorithms}, SODA '01, page 232–233, USA, 2001. Society for
  Industrial and Applied Mathematics.

\bibitem{chekuri_korula08}
Chandra Chekuri and Nitish Korula.
\newblock Single-sink network design with vertex connectivity requirements.
\newblock In Ramesh Hariharan, Madhavan Mukund, and V~Vinay, editors, {\em
  IARCS Annual Conference on Foundations of Software Technology and Theoretical
  Computer Science}, volume~2 of {\em Leibniz International Proceedings in
  Informatics (LIPIcs)}, pages 131--142, Dagstuhl, Germany, 2008. Schloss
  Dagstuhl--Leibniz-Zentrum fuer Informatik.
\newblock URL: \url{http://drops.dagstuhl.de/opus/volltexte/2008/1747}, \href
  {https://doi.org/10.4230/LIPIcs.FSTTCS.2008.1747}
  {\path{doi:10.4230/LIPIcs.FSTTCS.2008.1747}}.

\bibitem{ChekuriP05}
Chandra Chekuri and Martin Pal.
\newblock A recursive greedy algorithm for walks in directed graphs.
\newblock In {\em 46th annual IEEE symposium on foundations of computer science
  (FOCS'05)}, pages 245--253. IEEE, 2005.

\bibitem{cllz22}
Q.~Chen, B.~Laekhanukit, C.~Liao, and Y.~Zhang.
\newblock Survivable network design revisited: Group-connectivity.
\newblock In {\em 2022 IEEE 63rd Annual Symposium on Foundations of Computer
  Science (FOCS)}, pages 278--289, Los Alamitos, CA, USA, nov 2022. IEEE
  Computer Society.
\newblock URL:
  \url{https://doi.ieeecomputersociety.org/10.1109/FOCS54457.2022.00033}, \href
  {https://doi.org/10.1109/FOCS54457.2022.00033}
  {\path{doi:10.1109/FOCS54457.2022.00033}}.

\bibitem{diameter_bounds}
F.~R.~K. Chung and M.~R. Garey.
\newblock Diameter bounds for altered graphs.
\newblock {\em Journal of Graph Theory}, 8(4):511--534, 1984.
\newblock URL:
  \url{https://onlinelibrary.wiley.com/doi/abs/10.1002/jgt.3190080408}, \href
  {http://arxiv.org/abs/https://onlinelibrary.wiley.com/doi/pdf/10.1002/jgt.3190080408}
  {\path{arXiv:https://onlinelibrary.wiley.com/doi/pdf/10.1002/jgt.3190080408}},
  \href {https://doi.org/https://doi.org/10.1002/jgt.3190080408}
  {\path{doi:https://doi.org/10.1002/jgt.3190080408}}.

\bibitem{cgns08}
Julia Chuzhoy, Anupam Gupta, Joseph~(Seffi) Naor, and Amitabh Sinha.
\newblock On the approximability of some network design problems.
\newblock {\em ACM Trans. Algorithms}, 4(2), may 2008.
\newblock \href {https://doi.org/10.1145/1361192.1361200}
  {\path{doi:10.1145/1361192.1361200}}.

\bibitem{dahl98}
Geir Dahl.
\newblock The 2-hop spanning tree problem.
\newblock {\em Operations Research Letters}, 23(1):21--26, 1998.
\newblock URL:
  \url{https://www.sciencedirect.com/science/article/pii/S0167637798000297},
  \href {https://doi.org/https://doi.org/10.1016/S0167-6377(98)00029-7}
  {\path{doi:https://doi.org/10.1016/S0167-6377(98)00029-7}}.

\bibitem{deboeck_fortz18}
Jérôme De~Boeck and Bernard Fortz.
\newblock {Extended formulation for hop constrained distribution network
  configuration problems}.
\newblock {\em European Journal of Operational Research}, 265(2):488--502,
  2018.
\newblock URL:
  \url{https://ideas.repec.org/a/eee/ejores/v265y2018i2p488-502.html}, \href
  {https://doi.org/10.1016/j.ejor.2017.08.01}
  {\path{doi:10.1016/j.ejor.2017.08.01}}.

\bibitem{DiarrassoubaGMGP16}
I.~Diarrassouba, V.~Gabrel, A.~R. Mahjoub, L.~Gouveia, and P.~Pesneau.
\newblock Integer programming formulations for the k-edge-connected
  3-hop-constrained network design problem.
\newblock {\em Networks}, 67(2):148--169, 2016.
\newblock URL: \url{https://onlinelibrary.wiley.com/doi/abs/10.1002/net.21667},
  \href
  {http://arxiv.org/abs/https://onlinelibrary.wiley.com/doi/pdf/10.1002/net.21667}
  {\path{arXiv:https://onlinelibrary.wiley.com/doi/pdf/10.1002/net.21667}},
  \href {https://doi.org/https://doi.org/10.1002/net.21667}
  {\path{doi:https://doi.org/10.1002/net.21667}}.

\bibitem{Diarrassouba_Mahjoub_Almudahka_2024}
I.~Diarrassouba, A.~R. Mahjoub, and I.~M. Almudahka.
\newblock Optimization algorithms for the k edge-connected l-hop-constrained
  network design problem.
\newblock {\em Soft Computing}, February 2024.
\newblock \href {https://doi.org/10.1007/s00500-023-09541-7}
  {\path{doi:10.1007/s00500-023-09541-7}}.

\bibitem{dkr16}
Michael Dinitz, Guy Kortsarz, and Ran Raz.
\newblock Label cover instances with large girth and the hardness of
  approximating basic k-spanner.
\newblock {\em ACM Trans. Algorithms}, 12(2), dec 2016.
\newblock \href {https://doi.org/10.1145/2818375} {\path{doi:10.1145/2818375}}.

\bibitem{frt_tree}
Jittat Fakcharoenphol, Satish Rao, and Kunal Talwar.
\newblock A tight bound on approximating arbitrary metrics by tree metrics.
\newblock {\em Journal of Computer and System Sciences}, 69(3):485--497, 2004.
\newblock Special Issue on STOC 2003.
\newblock URL:
  \url{https://www.sciencedirect.com/science/article/pii/S0022000004000637},
  \href {https://doi.org/https://doi.org/10.1016/j.jcss.2004.04.011}
  {\path{doi:https://doi.org/10.1016/j.jcss.2004.04.011}}.

\bibitem{filtser22}
A.~Filtser.
\newblock Hop-constrained metric embeddings and their applications.
\newblock In {\em 2021 IEEE 62nd Annual Symposium on Foundations of Computer
  Science (FOCS)}, pages 492--503, Los Alamitos, CA, USA, feb 2022. IEEE
  Computer Society.
\newblock URL:
  \url{https://doi.ieeecomputersociety.org/10.1109/FOCS52979.2021.00056}, \href
  {https://doi.org/10.1109/FOCS52979.2021.00056}
  {\path{doi:10.1109/FOCS52979.2021.00056}}.

\bibitem{GargKR98}
Naveen Garg, Goran Konjevod, and Ramamoorthi Ravi.
\newblock A polylogarithmic approximation algorithm for the group steiner tree
  problem.
\newblock {\em Journal of Algorithms}, 37(1):66--84, 2000.
\newblock Preliminary version in Proc.\ of ACM-SIAM SODA 1998.

\bibitem{hop_congestion21}
Mohsen Ghaffari, Bernhard Haeupler, and Goran Zuzic.
\newblock Hop-constrained oblivious routing.
\newblock In {\em Proceedings of the 53rd Annual ACM SIGACT Symposium on Theory
  of Computing}, STOC 2021, page 1208–1220, New York, NY, USA, 2021.
  Association for Computing Machinery.
\newblock \href {https://doi.org/10.1145/3406325.3451098}
  {\path{doi:10.1145/3406325.3451098}}.

\bibitem{GhugeN22}
Rohan Ghuge and Viswanath Nagarajan.
\newblock Quasi-polynomial algorithms for submodular tree orienteering and
  directed network design problems.
\newblock {\em Mathematics of Operations Research}, 47(2):1612--1630, 2022.

\bibitem{gpsv03}
L.~Gouveia, P.~Patricio, A.F. de~Sousa, and R.~Valadas.
\newblock Mpls over wdm network design with packet level qos constraints based
  on ilp models.
\newblock In {\em IEEE INFOCOM 2003. Twenty-second Annual Joint Conference of
  the IEEE Computer and Communications Societies (IEEE Cat. No.03CH37428)},
  volume~1, pages 576--586 vol.1, 2003.
\newblock \href {https://doi.org/10.1109/INFCOM.2003.1208708}
  {\path{doi:10.1109/INFCOM.2003.1208708}}.

\bibitem{GouveiaML18}
Luis Gouveia, Martim Joyce-Moniz, and Markus Leitner.
\newblock Branch-and-cut methods for the network design problem with
  vulnerability constraints.
\newblock {\em Computers \& Operations Research}, 91:190--208, 2018.

\bibitem{GouveiaL17}
Luis Gouveia and Markus Leitner.
\newblock Design of survivable networks with vulnerability constraints.
\newblock {\em European Journal of Operational Research}, 258(1):89--103, 2017.
\newblock URL:
  \url{https://www.sciencedirect.com/science/article/pii/S0377221716307172},
  \href {https://doi.org/https://doi.org/10.1016/j.ejor.2016.09.003}
  {\path{doi:https://doi.org/10.1016/j.ejor.2016.09.003}}.

\bibitem{grandoni_italiano06}
Fabrizio Grandoni and Giuseppe~F. Italiano.
\newblock Improved approximation for single-sink buy-at-bulk.
\newblock In Tetsuo Asano, editor, {\em Algorithms and Computation}, pages
  111--120, Berlin, Heidelberg, 2006. Springer Berlin Heidelberg.

\bibitem{GrandoniLL19}
Fabrizio Grandoni, Bundit Laekhanukit, and Shi Li.
\newblock O (log2 k/log log k)-approximation algorithm for directed steiner
  tree: a tight quasi-polynomial-time algorithm.
\newblock In {\em Proceedings of the 51st Annual ACM SIGACT Symposium on Theory
  of Computing}, pages 253--264, 2019.

\bibitem{gmm01}
Sudipto Guha, Adam Meyerson, and Kamesh Munagala.
\newblock A constant factor approximation for the single sink edge installation
  problems.
\newblock In {\em Proceedings of the Thirty-Third Annual ACM Symposium on
  Theory of Computing}, STOC '01, page 383–388, New York, NY, USA, 2001.
  Association for Computing Machinery.
\newblock \href {https://doi.org/10.1145/380752.380827}
  {\path{doi:10.1145/380752.380827}}.

\bibitem{gkr10}
Anupam Gupta, Ravishankar Krishnaswamy, and R.~Ravi.
\newblock Tree embeddings for two-edge-connected network design.
\newblock In {\em Proceedings of the Twenty-First Annual ACM-SIAM Symposium on
  Discrete Algorithms}, SODA '10, page 1521–1538, USA, 2010. Society for
  Industrial and Applied Mathematics.

\bibitem{gmr03}
Anupam Gupta, Amit Kumar, and Tim Roughgarden.
\newblock Simpler and better approximation algorithms for network design.
\newblock In {\em Proceedings of the Thirty-Fifth Annual ACM Symposium on
  Theory of Computing}, STOC '03, page 365–372, New York, NY, USA, 2003.
  Association for Computing Machinery.
\newblock \href {https://doi.org/10.1145/780542.780597}
  {\path{doi:10.1145/780542.780597}}.

\bibitem{gnr10}
Anupam Gupta, Viswanath Nagarajan, and R~Ravi.
\newblock An improved approximation algorithm for requirement cut.
\newblock {\em Operations Research Letters}, 38(4):322--325, 2010.

\bibitem{hop_distance21}
Bernhard Haeupler, D.~Ellis Hershkowitz, and Goran Zuzic.
\newblock Tree embeddings for hop-constrained network design.
\newblock In {\em Proceedings of the 53rd Annual ACM SIGACT Symposium on Theory
  of Computing}, STOC 2021, page 356–369, New York, NY, USA, 2021.
  Association for Computing Machinery.
\newblock \href {https://doi.org/10.1145/3406325.3451053}
  {\path{doi:10.1145/3406325.3451053}}.

\bibitem{hks09}
Mohammad~Taghi Hajiaghayi, Guy Kortsarz, and Mohammad~R Salavatipour.
\newblock Approximating buy-at-bulk and shallow-light k-steiner trees.
\newblock {\em Algorithmica}, 53:89--103, 2009.

\bibitem{HalperinKKSW07}
Eran Halperin, Guy Kortsarz, Robert Krauthgamer, Aravind Srinivasan, and Nan
  Wang.
\newblock Integrality ratio for group steiner trees and directed steiner trees.
\newblock {\em SIAM Journal on Computing}, 36(5):1494--1511, 2007.

\bibitem{HalperinK03}
Eran Halperin and Robert Krauthgamer.
\newblock Polylogarithmic inapproximability.
\newblock In {\em Proceedings of the thirty-fifth annual ACM symposium on
  Theory of computing}, pages 585--594, 2003.

\bibitem{HuygensMP04}
David Huygens, Ali~Ridha Mahjoub, and Pierre Pesneau.
\newblock Two edge-disjoint hop-constrained paths and polyhedra.
\newblock {\em SIAM Journal on Discrete Mathematics}, 18(2):287--312, 2004.
\newblock \href
  {http://arxiv.org/abs/https://doi.org/10.1137/S0895480102419445}
  {\path{arXiv:https://doi.org/10.1137/S0895480102419445}}, \href
  {https://doi.org/10.1137/S0895480102419445}
  {\path{doi:10.1137/S0895480102419445}}.

\bibitem{Jain01}
K.~Jain.
\newblock A factor 2 approximation algorithm for the generalized {Steiner}
  network problem.
\newblock {\em Combinatorica}, 21(1):39--60, 2001.

\bibitem{kp06}
Erez Kantor and David Peleg.
\newblock Approximate hierarchical facility location and applications to the
  shallow steiner tree and range assignment problems.
\newblock In Tiziana Calamoneri, Irene Finocchi, and Giuseppe~F. Italiano,
  editors, {\em Algorithms and Complexity}, pages 211--222, Berlin, Heidelberg,
  2006. Springer Berlin Heidelberg.

\bibitem{ks11}
M.~Reza Khani and Mohammad~R. Salavatipour.
\newblock Improved approximations for buy-at-bulk and shallow-light $k$-steiner
  trees and $(k,2)$-subgraph.
\newblock {\em J. Comb. Optim.}, 31(2):669–685, feb 2016.
\newblock \href {https://doi.org/10.1007/s10878-014-9774-5}
  {\path{doi:10.1007/s10878-014-9774-5}}.

\bibitem{kls05}
Jochen K{\"o}nemann, Asaf Levin, and Amitabh Sinha.
\newblock Approximating the degree-bounded minimum diameter spanning tree
  problem.
\newblock {\em Algorithmica}, 41:117--129, 2005.

\bibitem{kortsarz_nutov11}
Guy Kortsarz and Zeev Nutov.
\newblock Approximating some network design problems with node costs.
\newblock {\em Theoretical Computer Science}, 412(35):4482--4492, 2011.
\newblock URL:
  \url{https://www.sciencedirect.com/science/article/pii/S0304397511003021},
  \href {https://doi.org/https://doi.org/10.1016/j.tcs.2011.04.013}
  {\path{doi:https://doi.org/10.1016/j.tcs.2011.04.013}}.

\bibitem{kp97}
Guy Kortsarz and David Peleg.
\newblock Approximating shallow-light trees.
\newblock In {\em Proceedings of the Eighth Annual ACM-SIAM Symposium on
  Discrete Algorithms}, SODA '97, page 103–110, USA, 1997. Society for
  Industrial and Applied Mathematics.

\bibitem{lcm99}
Larry~J. LeBlanc, Jerome Chifflet, and Philippe Mahey.
\newblock Packet routing in telecommunication networks with path and flow
  restrictions.
\newblock {\em INFORMS Journal on Computing}, 11(2):188--197, 1999.
\newblock \href {http://arxiv.org/abs/https://doi.org/10.1287/ijoc.11.2.188}
  {\path{arXiv:https://doi.org/10.1287/ijoc.11.2.188}}, \href
  {https://doi.org/10.1287/ijoc.11.2.188} {\path{doi:10.1287/ijoc.11.2.188}}.

\bibitem{Lhomme15}
Serge Lhomme.
\newblock {Vulnerability and resilience of ports and maritime networks to
  cascading failures and targeted attacks}.
\newblock In Routledge, editor, {\em {Maritime Networks. Spatial Structures and
  Time Dynamic}}. {Routledge}, October 2015.
\newblock URL: \url{https://hal.science/hal-01275157}.

\bibitem{MahjoubPSU19}
A.~Ridha Mahjoub, Michael Poss, Luidi Simonetti, and Eduardo Uchoa.
\newblock Distance transformation for network design problems.
\newblock {\em SIAM Journal on Optimization}, 29(2):1687--1713, 2019.
\newblock \href {http://arxiv.org/abs/https://doi.org/10.1137/16M1108261}
  {\path{arXiv:https://doi.org/10.1137/16M1108261}}, \href
  {https://doi.org/10.1137/16M1108261} {\path{doi:10.1137/16M1108261}}.

\bibitem{MahjoubSU13}
A.~Ridha Mahjoub, Luidi Simonetti, and Eduardo Uchoa.
\newblock Hop-level flow formulation for the survivable network design with hop
  constraints problem.
\newblock {\em Networks}, 61(2):171--179, 2013.
\newblock URL: \url{https://onlinelibrary.wiley.com/doi/abs/10.1002/net.21483},
  \href
  {http://arxiv.org/abs/https://onlinelibrary.wiley.com/doi/pdf/10.1002/net.21483}
  {\path{arXiv:https://onlinelibrary.wiley.com/doi/pdf/10.1002/net.21483}},
  \href {https://doi.org/https://doi.org/10.1002/net.21483}
  {\path{doi:https://doi.org/10.1002/net.21483}}.

\bibitem{marathe98}
Madhav~V Marathe, R~Ravi, Ravi Sundaram, S.S Ravi, Daniel~J Rosenkrantz, and
  Harry~B Hunt.
\newblock Bicriteria network design problems.
\newblock {\em Journal of Algorithms}, 28(1):142--171, 1998.
\newblock URL:
  \url{https://www.sciencedirect.com/science/article/pii/S0196677498909300},
  \href {https://doi.org/https://doi.org/10.1006/jagm.1998.0930}
  {\path{doi:https://doi.org/10.1006/jagm.1998.0930}}.

\bibitem{mmp08}
Adam Meyerson, Kamesh Munagala, and Serge Plotkin.
\newblock Cost-distance: Two metric network design.
\newblock {\em SIAM Journal on Computing}, 38(4):1648--1659, 2008.
\newblock \href {http://arxiv.org/abs/https://doi.org/10.1137/050629665}
  {\path{arXiv:https://doi.org/10.1137/050629665}}, \href
  {https://doi.org/10.1137/050629665} {\path{doi:10.1137/050629665}}.

\bibitem{monma_sheng86}
C.~Monma and Diane Sheng.
\newblock Backbone network design and performance analysis: A methodology for
  packet switching networks.
\newblock {\em IEEE Journal on Selected Areas in Communications},
  4(6):946--965, 1986.
\newblock \href {https://doi.org/10.1109/JSAC.1986.1146400}
  {\path{doi:10.1109/JSAC.1986.1146400}}.

\bibitem{pirkul_soni03}
Hasan Pirkul and Samit Soni.
\newblock New formulations and solution procedures for the hop constrained
  network design problem.
\newblock {\em European Journal of Operational Research}, 148(1):126--140,
  2003.
\newblock URL:
  \url{https://www.sciencedirect.com/science/article/pii/S0377221702003661},
  \href {https://doi.org/https://doi.org/10.1016/S0377-2217(02)00366-1}
  {\path{doi:https://doi.org/10.1016/S0377-2217(02)00366-1}}.

\bibitem{Racke08}
Harald R\"{a}cke.
\newblock Optimal hierarchical decompositions for congestion minimization in
  networks.
\newblock STOC '08, page 255–264, New York, NY, USA, 2008. Association for
  Computing Machinery.
\newblock \href {https://doi.org/10.1145/1374376.1374415}
  {\path{doi:10.1145/1374376.1374415}}.

\bibitem{ravi94}
R.~Ravi.
\newblock Rapid rumor ramification: approximating the minimum broadcast time.
\newblock In {\em Proceedings 35th Annual Symposium on Foundations of Computer
  Science}, pages 202--213, 1994.
\newblock \href {https://doi.org/10.1109/SFCS.1994.365693}
  {\path{doi:10.1109/SFCS.1994.365693}}.

\bibitem{ReichW89}
Gabriele Reich and Peter Widmayer.
\newblock Beyond steiner's problem: A vlsi oriented generalization.
\newblock In {\em International Workshop on Graph-theoretic Concepts in
  Computer Science}, pages 196--210. Springer, 1989.

\bibitem{raj12}
André Rossi, Alexis Aubry, and Mireille Jacomino.
\newblock Connectivity-and-hop-constrained design of electricity distribution
  networks.
\newblock {\em European Journal of Operational Research}, 218(1):48--57, 2012.
\newblock URL:
  \url{https://www.sciencedirect.com/science/article/pii/S0377221711009052},
  \href {https://doi.org/https://doi.org/10.1016/j.ejor.2011.10.006}
  {\path{doi:https://doi.org/10.1016/j.ejor.2011.10.006}}.

\bibitem{GoldsteinR71}
B.~Rothfarb and M.~Goldstein.
\newblock The one-terminal telpak problem.
\newblock {\em Operations Research}, 19(1):156--169, 1971.
\newblock \href {http://arxiv.org/abs/https://doi.org/10.1287/opre.19.1.156}
  {\path{arXiv:https://doi.org/10.1287/opre.19.1.156}}, \href
  {https://doi.org/10.1287/opre.19.1.156} {\path{doi:10.1287/opre.19.1.156}}.

\bibitem{batb97}
F~Sibel Salman, Joseph Cheriyan, R~Ravi, and Sairam Subramanian.
\newblock Buy-at-bulk network design: Approximating the single-sink edge
  installation problem.
\newblock In {\em Proceedings of the eighth annual ACM-SIAM symposium on
  Discrete algorithms}, pages 619--628, 1997.

\bibitem{talwar02}
Kunal Talwar.
\newblock The single-sink buy-at-bulk lp has constant integrality gap.
\newblock In William~J. Cook and Andreas~S. Schulz, editors, {\em Integer
  Programming and Combinatorial Optimization}, pages 475--486, Berlin,
  Heidelberg, 2002. Springer Berlin Heidelberg.

\bibitem{woolston_albin88}
Kathleen~A. Woolston and Susan~L. Albin.
\newblock The design of centralized networks with reliability and availability
  constraints.
\newblock {\em Computers and Operations Research}, 15(3):207--217, 1988.
\newblock URL:
  \url{https://www.sciencedirect.com/science/article/pii/0305054888900330},
  \href {https://doi.org/https://doi.org/10.1016/0305-0548(88)90033-0}
  {\path{doi:https://doi.org/10.1016/0305-0548(88)90033-0}}.

\bibitem{akgun_tansel11}
İbrahim Akgün and Barbaros~Ç. Tansel.
\newblock New formulations of the hop-constrained minimum spanning tree problem
  via miller-tucker-zemlin constraints.
\newblock {\em European Journal of Operational Research}, 212(2):263--276,
  2011.
\newblock URL:
  \url{https://www.sciencedirect.com/science/article/pii/S0377221711001056},
  \href {https://doi.org/https://doi.org/10.1016/j.ejor.2011.01.051}
  {\path{doi:https://doi.org/10.1016/j.ejor.2011.01.051}}.

\end{thebibliography}

\end{document}